\documentclass[a4paper,UKenglish,cleveref, autoref, thm-restate]{lipics-v2021}

\usepackage[capitalise]{cleveref}
\usepackage{thm-restate}
\usepackage{csquotes}
\usepackage{amsthm}
\usepackage{amsmath}
\usepackage{amsfonts}
\usepackage{faktor}
\usepackage{bm}
\usepackage{booktabs}
\usepackage[normalem]{ulem}
\usepackage{graphicx}
\usepackage{subcaption}
\usepackage{mathtools}
\usepackage{subcaption}
\usepackage{float}
\usepackage{capt-of} 
\usepackage[small, bold, full]{complexity}
\useunder{\uline}{\ul}{}

\usepackage{enumerate}
\usepackage{enumitem}
\newlist{steps}{enumerate}{1}
\setlist[steps, 1]{label = Step \arabic*.}
\usepackage{mdframed}

\usepackage[table]{xcolor}
\definecolor{gryffindor}{RGB}{220,0,1}
\definecolor{slytherin}{RGB}{26,121,42}
\definecolor{hufflepuff}{RGB}{236,185,57}
\definecolor{ravenclaw}{RGB}{14,26,164}

\newcommand{\Q}{\mathbb{Q}}
\newcommand{\RR}{\mathbb{R}}
\newcommand{\pp}{\mathfrak{p}}

\newcommand{\Z}{\mathbb{Z}}
\newcommand{\N}{\mathbb{N}}
\newcommand{\KK}{\mathbb{K}}
\newcommand{\LL}{\mathbb{L}}
\renewcommand{\O}{\mathcal{O}}
\renewcommand{\poly}{\mathsf{poly}}
\newcommand{\CC}{\mathbb{C}}
\newcommand{\LRS}[1]{#1} 

\newcommand{\Alg}{\smash{\overline{\Q}}}

\newcommand{\orbit}{\mathcal{B}}

\newcommand{\zero}{\mathcal{Z}}

\newcommand{\orbprob}{\textnormal{\textsc{Orbit}}}
\newcommand{\skolem}{\textnormal{\textsc{Skolem}}}
\newcommand{\simskol}{\textnormal{\textsc{SimSkolem}}}
\newclass{\EqSLP}{EqSLP}

\DeclareMathOperator{\supp}{supp} 
\DeclareMathOperator{\spn}{span} 
\numberwithin{equation}{section}

\usepackage[colorinlistoftodos,prependcaption,obeyFinal]{todonotes}

\title{On the Subspace Orbit Problem and the Simultaneous Skolem Problem}

\author{Piotr {Bacik}}{University of Oxford, UK \and Max Planck Institute for Software Systems, Saarland Informatics Campus, Germany}{piotr.bacik@stcatz.ox.ac.uk}{https://orcid.org/0009-0006-0248-3204}{Supported by EPSRC grant EP/X033813/1, ERC grant DynAMiCs (101167561) and DFG grant 389792660 as part
of TRR 248.}
\author{Anton Varonka}{TU Wien, Austria}{anton.varonka@tuwien.ac.at}{https://orcid.org/0000-0001-5758-0657}{Supported by the ERC consolidator grant ARTIST 101002685.\\
	\includegraphics[height=1em]{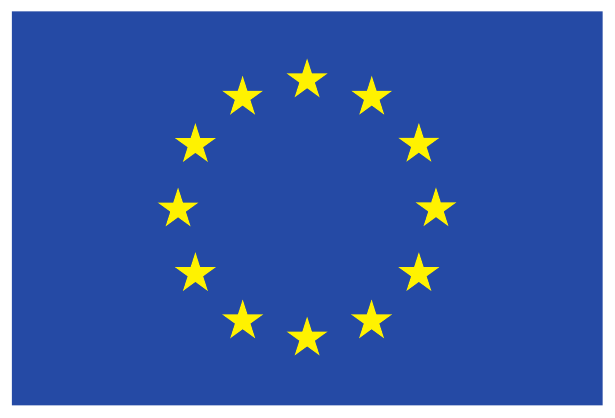} This paper is part of a project that has received funding from the European Research Council (ERC) under the European Union's Horizon 2020 research and innovation program (grant agreement No.~10103444).}

\Copyright{Piotr Bacik and Anton Varonka}
\authorrunning{P.\ Bacik and A.\ Varonka}

\nolinenumbers



\acknowledgements{The authors would like to thank Jo\"el Ouaknine and James Worrell for helpful comments. 
	The authors would like to express gratitude to the organisers of the SAMSA 2025 workshop, who provided a fertile ground for the ideas that led to this paper.}

\EventEditors{Claudia Faggian and Joost-Pieter Katoen}
\EventNoEds{2}
\EventLongTitle{41st Annual Symposium on Logic in Computer Science (LICS 2026)}
\EventShortTitle{LICS 2026}
\EventAcronym{LICS}
\EventYear{2026}
\EventDate{July 20--23, 2026}
\EventLocation{Lisbon, Portugal}
\EventLogo{}
\SeriesVolume{380}
\ArticleNo{49}

\begin{document}
	\maketitle

\begin{abstract}
	The Orbit Problem asks whether the orbit of a point under a matrix reaches a given target set.
	When the target is a single point, the problem was shown to be decidable in polynomial time
	by Kannan and Lipton. 
	This decidability result was later extended by Chonev et al.\ to targets of dimension~3 (in arbitrary ambient dimension), but decidability remains open for subspaces of dimension~4.  
	At the other extreme, the special case of the Orbit Problem in which the
	target set is a hyperplane of co-dimension~1 is equivalent to the Skolem Problem for linear recurrence sequences,
	whose decidability has been open for many decades.
	
	In this paper, we show that 
	the Orbit Problem is decidable if the target subspace has dimension logarithmic in the dimension of the orbit.
	Over the rationals, we moreover obtain a complexity bound $\NP^{\RP}$ in this case, when the target space dimension is bounded. 
	On the other hand, we show that 
	the version of the
	Orbit Problem 
	where the dimension of the target subspace is linear in the dimension of the orbit is as hard as the Skolem Problem. 
\end{abstract}

\section{Introduction}
\label{sec:intro}
A \emph{linear dynamical system} (LDS) consists of a matrix $A \in \KK^{d\times d}$ and initial vector $\bm x \in \KK^d$, where $\KK$ is a field. An LDS $(A,\bm x)$ defines an \emph{orbit}  $\orbit(A,\bm x) = \{\bm x, A \bm x, A^2 \bm{x}, \dots \}$ given by the image of $\bm x$ under repeated application of $A$. The Orbit Problem, introduced by Harrison~\cite{harrison_book1969} in the context of the reachability problem for linear sequential machines, asks whether a given LDS $(A,\bm x) \in \Q^{d \times d}\times \Q^d$ reaches a given point $\bm y \in \Q^d$, that is, whether $\bm y \in \orbit(A,\bm x)$. 
A decade later, this problem was shown to be decidable in polynomial time by Kannan and Lipton~\cite{KL_decid1980,KL_poly1986}. 

A natural generalisation to the problem where the target point~$\bm y$ is replaced by a target subspace $S \subseteq \Q^d$ was discussed in the journal version of their work~\cite{KL_poly1986}. 
%
\begin{mdframed}
	{\bf \textsc{Subspace Orbit Problem}.}\\
	Given an LDS $(A,\bm x) \in \KK^{d\times d} \times \KK^d$ and a subspace $S \subseteq \KK^d$, decide whether $\orbit(A,\bm x) \cap S \neq \varnothing$.
\end{mdframed}\vspace{1mm}
Around 30 years later, a significant step forward was achieved by Chonev, Ouaknine and Worrell~\cite{chonev_orbit2013,chonev_complex2016}, who showed the Subspace Orbit Problem is in $\NP^\RP$ when the dimension of the target subspace is at most~3, regardless of the dimension~$d$ of the ambient space.
To date, this is the state of the art for the Subspace Orbit Problem, with decidability for higher-dimensional target subspaces remaining a major open problem.

Other domains for targets have been studied in~\cite{chonev_polyhedra2015, almagor_2019, karimov_2020}, among which one finds polyhedra and more generally semi-algebraic sets. Cumulatively, we call these problems the \emph{orbit problems}. A plethora of further generalisations and variants of orbit problems may be found in the literature \cite{baier_2021, Almagor_2021,Karimov_2022,dcosta_2021,dcosta_2022}.
An important application of the orbit problems comes from program analysis:
note that LDS can be viewed as linear while loops 
\begin{equation*}
	\bm{v} \leftarrow \bm{x}; \enspace \textbf{while}\ \textsc{cond}(v)\ \textbf{do} \enspace \bm{v}\leftarrow A \cdot \bm{v},
\end{equation*}
where reachability of a target~$S$ corresponds to the halting problem with $\textsc{cond}(\bm{v}) = \bm{v}\not\in S$; that is, where the looping condition is expressed as a finite disjunction of linear disequalities.
We refer the reader to~\cite{Karimov_2022_lds,ouaknine_termination_2015,karimov_multiple_2025} for further discussions of LDS and loops.

Taking one step back, a major obstruction on the way to deciding orbit problems is the Skolem Problem. 
A \emph{linear recurrence relation} of order~$d$ over a field~$\KK$ is an equation
\begin{equation}\label{eq:recrel}
	u_n = c_1u_{n-1} + \dots + c_du_{n-d},
\end{equation}
where $c_i \in \KK$ for $1 \leq i \leq d$ and $c_d \neq 0$.
A linear recurrence sequence (LRS) $u = \langle u_n\rangle_{n\in\N}$ is a sequence of elements in~$\KK$ satisfying~\eqref{eq:recrel}.
In the case that $\KK$ has characteristic 0, the set~$\zero(u) = \{n \in \N : u_n = 0\}$ of zero terms of an LRS
is described by the theorem of Skolem, Mahler, and Lech \cite{Skolem_SML,Mahler_SML,lech_note_1953}:
\begin{theorem}[Skolem--Mahler--Lech] \label{thm:SML}
The zero set $\zero(u)$ is a union of finitely many arithmetic progressions and a finite set.
\end{theorem}
The proof is unfortunately ineffective --- though we can find the arithmetic progressions, there is no known algorithm to compute the exceptional finite set. Computing this finite set is equivalent to solving the \emph{Skolem Problem}.
\begin{mdframed}
	{\bf \textsc{Skolem Problem}.}\\
	Given an LRS~$u$, decide whether 
	$\zero(u)$	is non-empty.
\end{mdframed}
\vspace{1mm}
We denote the Skolem Problem on $\KK$-LRS that satisfy a recurrence relation of order~$d$ as $\skolem_\KK(d)$; when $\KK = \Alg$ we simply write $\skolem(d)$. It is straightforward to see that deciding the Skolem Problem is equivalent to deciding the Subspace Orbit Problem with a hyperplane target, that is, a linear subspace of $\KK^d$ of dimension~$d-1$. However, despite being posed close to a century ago, decidability of the Skolem Problem has stubbornly remained open; decidability is known for $\skolem(4)$ \cite{tijdeman_mst1984,vereshchagin_1985,bacik_completing_2025}, and furthermore, $\skolem_\Q(4)$ lies in $\coRP$ \cite{bacik_complexity_2026}, but $\skolem_\Q(5)$ remains open \cite{Bilu_2023}. In the orbit problems terminology, the decidability of reaching a 4-dimensional subspace of $\Q^5$ is open.

When considering \emph{non-degenerate} sequences, the statement of \cref{thm:SML} can be further refined. Define the \emph{characteristic polynomial} of $u$ satisfying \cref{eq:recrel} to be
\begin{align*}
    g(x) = x^d - c_1 x^{d-1} - \dots - c_d \, .
\end{align*}
We say $u$ is non-degenerate if no ratio of distinct roots of its characteristic polynomial is a root of unity.
The Skolem--Mahler--Lech Theorem shows that a non-degenerate LRS is either 
identically zero or has finitely many zero terms.
Crucially, for every LRS~$\langle u_n \rangle_{n\in\N}$ there exists an integer~$L$ such that
each subsequence $\langle u_{Ln + r} \rangle_{n\in \N}$ with $0 \leq r \leq L-1$
is 
non-degenerate.
%
\subsection*{The Reduced Subspace Orbit Problem}
The Subspace Orbit Problem with target space $S \subseteq \Q^5$ of dimension~4 is equivalent to $\skolem_\Q(5)$, the decidability of which appears to be beyond the reach of current methods. A natural question is to ask whether the problem with target dimension~4 becomes tractable when the dimension of the LDS is sufficiently large. However, for this we need a sufficiently robust notion of the dimension of an LDS. For example, if $(A,\bm x) \in \KK^{d \times d} \times \KK^d$, then one might naïvely take the dimension of $(A,\bm x)$ to be the ambient dimension~$d$. We argue against it, as any LDS may trivially be viewed as being inside a space of arbitrarily large dimension by adding new components to $A,\bm{x}$ that are zero everywhere. 
Instead, we introduce the notion of \emph{reduced} LDS, and of the \emph{inherent} dimension of an LDS.

Call a LDS $(A,\bm x)$ \emph{non-degenerate} if no ratio of distinct eigenvalues of the matrix $A$ is a root of unity. Define the \emph{Krylov dimension} of an orbit to be the smallest dimension of any subspace that contains it. We say an LDS $(A, \bm x) \in \Alg^{d \times d} \times \Alg^d$ is \emph{reduced} if it is non-degenerate and the Krylov dimension of any tail of its orbit is $d$; that is, $\dim \spn\{A^n \bm x : n \geq m \} = d$ for all $m \geq 0$.



By removing finitely many initial points, taking sub-orbits and restricting to smaller subspaces (see \cref{sec:prelim}), the orbit of any LDS may be rewritten as a union of finitely many points, and finitely many orbits of reduced LDS with equal Krylov dimensions. We define the \emph{inherent dimension} of an LDS to be the Krylov dimension of its associated reduced LDS.\footnote{Reduced LDS have the property that their Krylov dimension remains the same upon removing finitely many points or taking sub-orbits, so the notion of inherent dimension is well defined.} The inherent dimension may be more intuitively characterised by the following property: if an LDS has inherent dimension $d$ then its orbit cannot be contained within a finite union of dimension $d-1$ subspaces.

It turns out that the notion of inherent dimension is the right one for investigating the problem at higher dimensions --- although deciding the Subspace Orbit Problem for LDS of inherent dimension~$d$ and target subspaces of dimension~$t$ is $\skolem(5)$-hard when $(d,t) = (5,4)$, 
the problem becomes tractable already when $(d,t) = (6,4)$. This motivates the following question, which this paper investigates: {\bf for each fixed target dimension, what is the smallest inherent dimension of LDS for which we can decide the Subspace Orbit Problem?}

The Subspace Orbit Problem on arbitrary LDS can be solved by solving the Subspace Orbit Problem on each associated reduced LDS. 
Therefore, in the rest of the paper we assume that \textbf{any LDS $(A, \bm x)$ is reduced}, unless stated otherwise. 
We denote the Subspace Orbit Problem restricted to reduced LDS $(A,\bm x) \in \KK^{d \times d} \times \KK^d$ and subspaces $S \subseteq \KK^d$ of dimension $t$ by $\orbprob_\KK(d,t)$. 

We investigate the central question via the Simultaneous Skolem Problem, also discussed in~\cite{bacik_p-adic_2025}. \\

		\begin{mdframed}
			{\bf \textsc{Simultaneous Skolem Problem}.}\\
			Given a finite set of LRS $u^{(1)}, \dots, u^{(k)}$, decide whether 
			\[\zero(u^{(1)}, \dots, u^{(k)}) = \{n \in \N : u^{(i)}_n = 0 \text{ for all } 1 \leq i \leq k\}\]
			is non-empty.
		\end{mdframed}\vspace{1mm}
Moreover, we are interested in computing the set $\zero(u^{(1)}, \dots, u^{(k)})$. While the Skolem Problem corresponds to the Subspace Orbit Problem where targets are hyperplanes, the Simultaneous Skolem Problem can express any subspace target. 

We further denote the Simultaneous Skolem Problem on $k$ linearly independent\footnote{Say $u^{(1)}, \dots, u^{(k)}$ are linearly dependent if there exist $b_1, \dots, b_k \in \Alg$ such that $b_1u^{(1)}_n + \dots + b_k u^{(k)}_n = 0$, for all $n \in \N$.} non-degenerate $\KK$-LRS satisfying the same recurrence relation of order $d$ by $\simskol_\KK(d,k)$.

Just like how the Subspace Orbit Problem for general LDS reduces to considering reduced LDS, the Simultaneous Skolem Problem for general LRS reduces to linearly independent and non-degenerate LRS by passing to non-degenerate subsequences and removing linearly dependent LRS. 
We shall see that while known algorithms only yield~$\zero(u)$ for~$u$ of order at most~4, considering $k\geq 2$ LRS simultaneously renders decidable problems for LRS of higher orders.
\subsection*{Overview of the results}

When $\KK = \Alg$, we omit the subscript, i.e.\ $\orbprob_{\Alg}(d,t) = \orbprob(d,t)$ and $\simskol_{\Alg}(d,k) = \simskol(d,k)$.


We first show that deciding $\orbprob(d,t)$ is equivalent to deciding $\simskol(d,d-t)$.\footnote{Note that in conjunction with the simple observation that $\simskol(d,k)$ reduces to $\simskol(d-1,k-1)$ (see \cref{lem:growambient}) and the fact that $\skolem(4)$ is decidable~\cite{bacik_completing_2025}, this already extends the decidability results of Chonev et.\ al.\ \cite{chonev_complex2016} from vector spaces over~$\Q$ to vector spaces over~$\Alg$.} 
We then study $\simskol(d,k)$, in which the key observation is that all simultaneous zeros of $u^{(1)}, \dots, u^{(k)}$ are zeros of an arbitrary linear combination~$w = \sum_{i=1}^k \alpha_i u^{(i)}$ of these LRS\@.
$\simskol$ may thus be solved by solving the Skolem Problem for $\langle w_n\rangle_{n\in\N}$.
More precisely, provided that $u^{(1)}, \dots, u^{(k)}$ are non-degenerate and linearly independent,
any non-trivial linear combination~$w$ of these LRS is a non-zero non-degenerate LRS. 
By \cref{thm:SML}, the zero set~$\zero(w)$ is finite; if we can compute it exactly, all that is left to do in order to compute~$\zero(u^{(1)}, \dots, u^{(k)})$ is to check whether $u_{n}^{(i)} = 0$ for each $n \in \zero(w)$ and all~$1\leq i \leq k$. The freedom to choose the linear combination~$w$ plays a crucial role in achieving our decidability results.
Using this, we show decidability of
\begin{itemize}
    \item $\simskol(6,2)$ and $\orbprob(6,4)$,
    \item $\simskol(9,4)$ and $\orbprob(9,5)$,
    \item $\simskol(12,6)$ and $\orbprob(12,6)$.
\end{itemize}
We provide examples that show that these results are the best possible with the techniques we use, even when restricting to $\simskol_\Q(d,k)$ and $\orbprob_\Q(d,t)$. 

In addition to these low-dimensional results, we show that 
{\bf for each target dimension, the Subspace Orbit Problem is decidable whenever the inherent dimension of the LDS is sufficiently large.}
In particular, we show the following.
\begin{restatable}{theorem}{mainresult} \label{thm:main_result}
$\simskol(d,k)$ is decidable for all $(d,k)$ satisfying $d-k \leq 2\log_3 d$. Equivalently, $\orbprob(d,t)$ is decidable whenever $t \leq 2\log_3 d$.
\end{restatable}
We provide a careful complexity analysis which shows that, for any fixed $T$, the problem $\orbprob_\Q(d,t)$ for arbitrary $d$ and $t \leq T$ satisfying $t \leq 2\log_3 d$ is in $\NP^\RP$, and if $d$ is fixed then the problem lies in $\coRP$.

Finally, we provide an improved reduction of the Skolem Problem to the Subspace Orbit Problem. Specifically, we show that if there existed a constant $C \in (0,1)$ such that there was an algorithm solving $\orbprob(d,t)$ for all $t \leq Cd$, then such an algorithm could be used to decide the Skolem Problem.

Our contributions are visually summarised in~\cref{table}.
\begin{figure*}[h!]
	\setlength{\arrayrulewidth}{0.2pt}
	\small
	\captionsetup{font=small}
	\resizebox{\textwidth}{!}{%
	\begin{tabular}{|c|c|c|c|c|c|c|c|c|c|c|c|c|c|}
		\hline
		& $d=1$ & $d=2$ & $d=3$ & $d=4$ & $d=5$ & $d=6$
		& $d=7$ & $d=8$ & $d=9$ & $d=10$ & $d=11$ & $d=12$ & $\cdots$ \\
		\hline
		
		$t=1$ &
		\cellcolor{gray!30} & \multicolumn{12}{c|}{\cellcolor{slytherin!30}}\\
		\hline
		
		$t=2$ &
		\cellcolor{gray!30} & \cellcolor{gray!30} & \multicolumn{11}{c|}{\cellcolor{slytherin!30}\centering decid.\ by Chonev et al.\ (over $\mathbb{Q}$)} \\
		\hline
		
		$t=3$ &
		\cellcolor{gray!30} & \cellcolor{gray!30} & \cellcolor{gray!30} & \multicolumn{10}{c|}{\cellcolor{slytherin!30}}\\
		\hline
		

		$t=4$ &
		\cellcolor{gray!30} & \cellcolor{gray!30} & \cellcolor{gray!30} & \cellcolor{gray!30} & \cellcolor{gryffindor!20}$\skolem(5)$ & \cellcolor{ravenclaw!30}decid.& \cellcolor{ravenclaw!30} $\rightarrow$& \cellcolor{ravenclaw!30}&\cellcolor{ravenclaw!30} &\cellcolor{ravenclaw!30} & \cellcolor{ravenclaw!30}& \cellcolor{ravenclaw!30}&\cellcolor{ravenclaw!30} \dots \\
		\hline
		
		$t=5$ &
		\cellcolor{gray!30} & \cellcolor{gray!30} & \cellcolor{gray!30} & \cellcolor{gray!30} &
		\cellcolor{gray!30} &\cellcolor{gryffindor!30}$\skolem(6)$ & \cellcolor{yellow!50} & \cellcolor{yellow!50} &\cellcolor{ravenclaw!30}decid. &\cellcolor{ravenclaw!30} $\rightarrow$& \cellcolor{ravenclaw!30}&\cellcolor{ravenclaw!30} & \cellcolor{ravenclaw!30} \dots \\
		\hline
		
		$t=6$ &
		\cellcolor{gray!30} & \cellcolor{gray!30} & \cellcolor{gray!30} & \cellcolor{gray!30} &
		\cellcolor{gray!30} & \cellcolor{gray!30} & \cellcolor{gryffindor!40}$\skolem(7)$ & \cellcolor{gryffindor!20}$\skolem_\Q(5)$ & \cellcolor{yellow!50} & \cellcolor{yellow!50} & \cellcolor{yellow!50} &\cellcolor{ravenclaw!30} decid. & \cellcolor{ravenclaw!30} $\rightarrow$\\
		\hline
		
		$\vdots$ &
		\cellcolor{gray!30} & 	\cellcolor{gray!30} & 	\cellcolor{gray!30} & 	\cellcolor{gray!30} & 	\cellcolor{gray!30} & 	\cellcolor{gray!30} & 	\cellcolor{gray!30}
		& \cellcolor{gryffindor!50}$\ddots$ & \vdots & \vdots & \vdots & \vdots & $\ddots$ \\
		\hline
	\end{tabular}
}
\caption{A table depicting the state of the art, and our contributions for $\orbprob_{\Alg}(d,t)$. \newline
	\textcolor{slytherin}{Green}: shown decidable over $\Q$ by \cite{chonev_complex2016}, and extended in the present article to $\Alg$. \newline
	\textcolor{ravenclaw}{Blue}: instances newly shown decidable in the present paper over~$\Alg$; we show $\orbprob_\Q(d,t)$ for them to be in $\NP^\RP$, and in $\coRP$ when~$d$ is fixed. \newline
	\textcolor{yellow}{Yellow}: instances which are provably not MSTV-reducible (cf.~\cref{def:MSTV-reducible}), as shown by~\cref{ex:notmstv}.\newline
	\textcolor{gryffindor}{Red}: instances which are $\skolem_\Q(k)$-hard or $\skolem(k)$-hard for some~$k$ (see Appendix \ref{sec:table_hard}).
}
	\label{table}
\end{figure*}
\section{Preliminaries}\label{sec:prelim}
\subsection*{Reducing LDS}
The \emph{Krylov subspace} $V \subseteq \Alg^d$ associated to an LDS $(A, \bm x) \in \Alg^{d \times d} \times \Alg^d$ is the $\Alg$-vector space spanned by the orbit of~$\bm{x}$ under the powers of~$A \in \Alg^{d\times d}$:
\[V\coloneq\spn(\bm{x}, A\bm{x}, A^2\bm{x}, \dots).\]
Let $V_k \coloneq \spn(\bm{x}, A\bm{x}, A^2\bm{x}, \dots, A^{k-1}\bm{x})$ be a vector subspace spanned by the first~$k$ vectors of the orbit. 
By definition, let $V_0 \coloneqq \{\bm{0}\}$.
Clearly, \begin{equation}\label{eq:subspace-chain}
	V_0 \subseteq V_1 \subseteq \dots \subseteq V_k \subseteq \dots \subseteq V.
\end{equation}
Let~$\mu \leq d$ be the minimum power such that $A^\mu\bm{x} \in V_\mu$.
Then the increasing chain~\eqref{eq:subspace-chain} stabilises at~$k = \mu$, that is, $V_k = V_{k+1} = \dots = V$.
Moreover, $\mu$ is the dimension of the Krylov subspace~$V$ which we also call the \emph{Krylov dimension} of~$V$. 
The use of the eponym is a tribute to the seminal paper~\cite{krylov_1931} by Krylov;
for similar usage see e.g.~\cite[Part~VI]{trefethen_numerical_1997}. 
We write $\dim(A, \bm x) \coloneq \mu$. If $\mu = d$ we say $(A, \bm x)$ is \emph{full-dimensional}.
We say $(A, \bm x)$ has \emph{stable dimension} if $\dim(A, \bm x) = \dim(A, A^n \bm x)$ for all $n \in \N$. Observe that $(A,\bm x)$ certainly has stable dimension if $A$ is invertible. The following lemma shows how we can reduce to considering invertible $A$, and that when taking sub-orbits, each sub-orbit has equal Krylov dimension. See also \cite[Section 6.1]{karimov_thesis} for similar discussion.

\begin{lemma}\label{lem:stable}
%
	Let $(A,\bm x) \in \Alg^{d \times d} \times \Alg^d$ be an LDS and let $L$ be an integer. 
	Then there exists an integer~$\mu$ and an invertible matrix~$M \in \Alg^{d\times d}$ 
	(both depending only on $A$, $\bm{x}$ and $L$) 
	such that for all $ 0 \leq r \leq L-1$, the LDS $(A^L, A^{d+r}\bm x)$ has stable dimension $\mu$ and $\orbit(A^L, A^{d+r}\bm{x}) = \orbit(M, A^{d+r}\bm{x})$.
\end{lemma}
\begin{proof}
First, note that if $A$ is nilpotent then $A^n = 0$ for all $n \geq d$, and therefore the conclusion of the lemma follows trivially, with $\mu = 0$ and $M = I_d$, the $d \times d$ identity matrix. Henceforth assume $A$ is not nilpotent, and so has some non-zero eigenvalue. Let $J = PAP^{-1}$ be the Jordan normal form with block-diagonal such that $J = \mathrm{diag}(J_1,J_2)$ with $J_1 \in \Alg^{d_1 \times d_1}$ invertible and $J_2 \in \Alg^{d_2 \times d_2}$ nilpotent, with $d_1 > 0$. 
Let $\bm s = P \bm x = (\tilde{\bm{s}}_1, \tilde{\bm{s}}_2)^\top$, for $\tilde{\bm{s}}_1 \in \Alg^{d_1}$, $\tilde{\bm{s}}_2 \in \Alg^{d_2}$. Then, since $J_2^n = 0$ for all $n \geq d$, we have $J^n \bm s = (J_1^n \tilde{\bm{s}}_1,\bm{0})^\top$ for all $n \geq d$. 

Define $\tilde J = \mathrm{diag}(J_1,I_{d_2})$. 
Since $J^n \bm s = (J_1^n \tilde{\bm{s}}_1,\bm{0})^\top$ for all $n \geq d$, we have $J^n J^{d} \bm s = \tilde J^{n} J^{d} \bm s$ for all $n \geq 0$. 
That is, $\orbit(\tilde{J},J^{d}\bm{s}) = \orbit(J,J^{d}\bm{s})$.

Fix an integer~$L$. Since $J_1$ is invertible, $\dim(J_1^L, J_1^d\tilde{\bm{s}}_1)$ is stable; define~$\mu\coloneq\dim(J_1^L, J_1^d\tilde{\bm{s}}_1)$.
For each $0 \leq r \leq L-1$, we have \[\tilde J^{r} \orbit(\tilde J^{L}, J^d \bm s) = \orbit(\tilde J^{L} , J^{d+r} \bm s).\]
The Krylov subspaces of $(\tilde J^{L}, J^{d+r} \bm s)$, $0\leq r\leq L-1$, are related by multiplication by an invertible matrix. It follows that \[\dim({\tilde J}^{L}, J^d\bm{s}) = \dim({\tilde J}^{L}, J^{d+1}\bm{s}) = \dots = \dim({\tilde J}^{L}, J^{d+L-1}\bm{s}) =  \mu.\]

Recall that~$P$ is invertible.
By letting $M = P^{-1}\tilde J^{L} P$, it follows that that $(M, A^{d+r}\bm{x})$ has stable dimension~$\mu$ for each $0 \leq r \leq L-1$. Note also that $M$ is invertible, and that $\orbit(A^L,A^{d+r} \bm x) = \orbit(M, A^{d+r} \bm x)$, proving what was desired.
\end{proof}
\begin{remark} \label{rmk:mu_jordan_0}
From the proof, we see that $\mu \leq d-d_2$, where $d_2$ is the multiplicity of~$0$ as a root of the characteristic polynomial of~$A$.
\end{remark}

Recall that an LDS $(A, \bm x)$ is \emph{reduced} if it is non-degenerate, full-dimensional and has stable dimension. Note that if $(A,\bm x)$ is reduced then $A$ is invertible. Indeed, if not then $0$ is an eigenvalue of $A$ and by \cref{rmk:mu_jordan_0} we have $\mu < d$. Recall also that $\orbprob(d,t)$ is the Subspace Orbit Problem on \emph{reduced} LDS $(A , \bm x) \subseteq \Alg^{d \times d} \times \Alg^d$ and subspace $S \subseteq \Alg^d$ of dimension $t$. We now formalise a reduction from an arbitrary instance $(A,\bm{x}, S)$ of the Subspace Orbit Problem for a non-reduced LDS $(A,\bm{x}) \in \Alg^{d\times d} \times \Alg^d$ to $\orbprob$. 

First, check whether $A^i\bm{x} \in S$ for some $i < d$. 
If not the case, find an integer~$L$ such that $A^L$ is non-degenerate.\footnote{We may take $L$ to be the lowest common multiple of the orders of all ratios of eigenvalues $\lambda_i/\lambda_j$ that are roots of unity. This is at most exponential in the degree of the number field containing all eigenvalues \cite[Section A.5]{chonev_complex2016}.} 
We apply~\cref{lem:stable} to find an invertible matrix~$M$ such that $\orbit(A^L,A^{d+r}\bm{x}) = \orbit(M, A^{d+r}\bm{x})$ for all $0 \leq r \leq L-1$. 
Now decompose the orbit $\orbit(A,A^d\bm{x})$ into~$L$ sub-orbits $\orbit_r = \orbit(M,A^{d+r}\bm{x})$, where $0 \leq r \leq L-1$. 
Moreover, for each~$r$, define $V_r \coloneq \spn(\orbit_r)$. 
These subspaces are all of dimension~$\mu \leq d$ by~\cref{lem:stable}, 
and so can be identified with $\Alg^\mu$. Recall that we call $\mu$ the \emph{inherent dimension} of $(A,\bm x)$. 
The new targets $S_r \coloneq S \cap V_r$ are linear subspaces of~$V_r$, respectively. 

The original orbit $\orbit(A,\bm x)$ reaches subspace~$S$ of dimension~$t$ if and only if at least one of the following holds: 
1) $A^i\bm{x} \in S$ for some $i < d$; 
2) for some~$r$, the orbit~$\orbit_r$ of the reduced LDS (in ambient dimension~$\mu \leq d$) reaches the subspace~$S_r \subseteq \Alg^\mu$ of dimension~$t_r \leq t$. Now, 1) can be solved by a finite check, while 2) consists of instances of $\orbprob(\mu,t_r)$ for each $0 \leq r \leq L-1$. We summarise this in the following statement.
\begin{proposition}
Solving the Subspace Orbit Problem on arbitrary LDS $(A,\bm x) \in \Alg^{d \times d} \times \Alg^d$ of inherent dimension $\mu$ and target subspace $S \subseteq \Alg^d$ of dimension $t$ reduces to solving finitely many instances of $\orbprob(\mu,t')$ for $t' \leq t$.
\end{proposition}

\subsection*{Algebraic Number Theory}
We recall some algebraic number theory. More details can be found in \cite{neukirch_algebraic_1999,Waldschmidt_book}.
Let $\KK$ be a finite extension of $\Q$, i.e, a number field. We define the notions of Archimedean and non-Archimedean absolute values on $\KK$, which are central to the decidability results we prove. 

When $p \in \Z$ is a prime, we define the $p$-adic valuation $v_p : \Q \to \Z \cup \{\infty\}$ by $v_p(0) = \infty$ and $v_p\left(p^r \frac{a}{b} \right) = r$, where $a,b,p \in \Z$ are pairwise coprime. We define the $p$-adic absolute value by $|x|_p = p^{-v_p(x)}$. We require a generalisation of this to number fields, via a notion of valuation with respect to a prime ideal. 

Let $\O_\KK$ denote the ring of algebraic integers in $\KK$. Define a \emph{fractional ideal} of $\KK$ to be a non-zero finitely generated $\O_\KK$-submodule of $\KK$. Equivalently $I$ is a fractional ideal if and only if there is $c \in \KK$ such that $cI \subseteq \O_\KK$ is an ideal of $\O_\KK$. The fractional ideals of $\KK$ form a group under multiplication. Any fractional ideal $I$ of $\KK$ has a unique decomposition \cite[p. 22]{neukirch_algebraic_1999}
\begin{align} \label{eqn:ideal_fac}
I = \prod_{i=1}^t \pp_i^{n_i}
\end{align}
where $t \in \N$, each $\pp_i \subseteq \O_\KK$ is a prime ideal and $n_i \in \Z$. 

For a prime ideal $\pp \subseteq \O_\KK$ we define the \emph{valuation} $v_\pp : \KK \to \Z \cup \{\infty\}$ by $v_\pp(0) = \infty$ and for non-zero $a \in \KK$, we define $v_\pp(a)$ to be the exponent of $\pp$ in the prime decomposition of the fractional ideal $a\O_\KK$ given by \eqref{eqn:ideal_fac}. There is exactly one prime $p \in \Z$ such that $v_\pp(p) > 0$; we say $\pp$ \emph{lies above} $p$, and define the ramification index $e_\pp = v_\pp(p)$. We define the $\pp$-adic absolute value $|x|_\pp = p^{-v_\pp(x)/e_\pp}$.

By Ostrowski's theorem, any (non-trivial) absolute value on $\Q$ is equivalent to the modulus $|\cdot|$ or a $p$-adic absolute value $|\cdot|_p$ for some prime $p \in \Z$. 
Furthermore, if $|\cdot|_v$ is a (non-trivial) absolute value on $\KK$, then its restriction to $\Q$ is either equivalent to the modulus $|\cdot|$ (in which case $|\cdot|_v$ is \emph{Archimedean} and we write $v \mid \infty$) or some $p$-adic absolute value $|\cdot|_p$ (in which case $|\cdot|_v$ is \emph{non-Archimedean} and we write $v \mid p$, or $v\nmid \infty$). We normalise $|\cdot|_v$ such that, on $\Q$, $|\cdot|_v$ coincides with the modulus $|\cdot|$ or $|\cdot|_p$ for some $p$; in the Archimedean case we have $|x|_v = x$ for all $x \in \Q_{\geq 0}$ and in the non-Archimedean case $|p|_v = p^{-1}$. Denote the set of non-trivial absolute values on $\KK$, normalised as above, by $M_\KK$.

If $v \in M_\KK$ is Archimedean, it either corresponds to a real embedding $\KK \overset{\sigma}\hookrightarrow \mathbb R$ or a pair of complex embeddings $\KK \overset{\sigma,\overline \sigma}\hookrightarrow \mathbb C$. In both cases, under the normalisation above we have $|x|_v = |x|_\sigma = |\sigma(x)|$. Define the \emph{local degree} $D_v$ of $v$ as $D_v = 1$ or $D_v = 2$ if $v$ corresponds to a real embedding or a pair of complex conjugate embeddings respectively.

If $v \in M_\KK$ is non-Archimedean, and $v \mid p$, then $v$ corresponds to a prime ideal $\pp \subseteq \O_\KK$ dividing $p$. Under the normalisation above, we have $|x|_v = |x|_\pp = p^{-v_\pp(x)/e_\pp}$. Define $D_v = [\KK_v:\Q_p]$ where $\KK_v,\Q_p$ are the completions of $\KK, \Q$ with respect to $|\cdot|_v, |\cdot|_p$ respectively.
Note also that any non-Archimedean $v\in M_\KK$ satisfies the ultrametric inequality: for any $x,y \in \KK$ we have $|x+y|_v \leq \max\{|x|_v,|y|_v\}$.

In our complexity analysis, we require a robust notion of the size of an algebraic number. This is the content of the next definition.
\begin{definition}
Define $\log^+(x) = \log \max\{1,x\}$. If $x \in \KK$ and $[\KK:\Q] = D$, we define the \emph{absolute logarithmic height} of $x$ (or just \emph{height} for short) as
\begin{align*}
    h(x) := \frac{1}{D} \sum_{v \in M_\KK} D_v \log^+ |x|_v \, .
\end{align*}
\end{definition}

It is known that the height depends only on $x$, not the choice of number field $\KK$. The height satisfies the following properties, found in \cite[Sections 3.2, 3.5]{Waldschmidt_book}.
\begin{proposition} \label{prop:h_props}
For any algebraic numbers $\alpha_1,\alpha_2, \dots, \alpha_m$ with $\alpha_1 \neq 0$, for any $n \in \Z$ and any absolute value $|\cdot|_v$ on $\Q(\alpha_1)$ we have
\begin{enumerate}
    \item $h(\alpha_1 \alpha_2) \leq h(\alpha_1) + h(\alpha_2)$, \label{enum:h_props_mult}
    \item $h(\alpha_1 + \dots + \alpha_m) \leq \log m + h(\alpha_1) + \dots + h(\alpha_m)$, \label{enum:h_props_add}
    \item $h(\alpha_1^n) = |n|h(\alpha_1)$, \label{enum:h_props_pow}
    \item $-[\Q(\alpha_1):\Q]h(\alpha_1) \leq \log |\alpha_1|_v \leq [\Q(\alpha_1):\Q]h(\alpha_1)$ \label{enum:h_props_log}
\end{enumerate}
\end{proposition}
\subsection*{Linear Recurrence Sequences}
Let $u$ be a $\Alg$-LRS satisfying recurrence relation \eqref{eq:recrel} of order $d$. Say $u$ has order $d$ if it does not satisfy any recurrence relation of order $d-1$. Note then $\skolem(d)$ refers to the Skolem problem for all LRS of order $\leq d$.

A $\Alg$-LRS $u$ satisfying \eqref{eq:recrel} has an exponential polynomial form
\begin{align} \label{eq:exp-poly}
    u_n = \sum_{i=1}^s P_i(n) \lambda_i^n = \sum_{i=1}^s \sum_{j=0}^{m_i-1} c_{i,j} n^j\lambda_i^n
\end{align}
where each $P_i(n) = \sum_{j=0}^{m_i-1} c_{i,j} n^j \in \Alg[n]$ is a polynomial, and $\lambda_1, \dots, \lambda_s$ are the roots of the \emph{characteristic polynomial}
$g(X) = X^d - c_1 X^{d-1} - \dots - c_d$, with multiplicities $m_1, \dots, m_s$. 
We call $\lambda_1, \dots, \lambda_s$ the \emph{characteristic roots} of $u$. If $\lambda_i/\lambda_j$ is a root of unity for any $i \neq j$ then we say $u$ is degenerate. For $|\cdot|_v$ an absolute value, call $\lambda_i$ \emph{dominant with respect to} $|\cdot|_v$ (or $|\cdot|_v$-dominant), if $|\lambda_i|_v \geq |\lambda_j|_v$ for all $j \neq i$. 
We say an LRS is \emph{simple} if $\deg P_i = 0$ for all~$i \in \{1, \dots, s\}$.

We make some further definitions to discuss different parts of the LRS. Precisely, a \emph{term} is a product of an exponential monomial $n^j\lambda_i^n$ and a non-zero constant~$c_{i,j}$. If $\lambda_i$ is dominant with respect to $|\cdot|_v$, then we call $c_{i,j}n^j \lambda_i^n$ a \emph{dominant term}. 
Fix a linear order on the set of exponential monomials in~\eqref{eq:exp-poly} and associate with~$u$ a vector $\bm{u} \in \Alg^d$ of constants $c_{i,j}$ with respect to this order.\footnote{Note that it is possible for the $c_{i,j}$ here to be zero, e.g. if the order of $u$ is $< d$.}
We refer to the set of non-zero components of a vector~$\bm{u}$ as its \emph{support}~$\supp(\bm{u})$.

Given~$k$ LRS $u^{(1)}, \dots, u^{(k)}$ satisfying \eqref{eq:recrel},
consider their associated vectors $\bm{u}^{(1)}, \dots, \bm{u}^{(k)} \in \Alg^d$.
We say that $u^{(1)}, \dots, u^{(k)}$ are \emph{linearly independent} over~$\Alg$ if 
the subspace~$C \subseteq \Alg^d$ spanned by $\bm{u}^{(1)}, \dots, \bm{u}^{(k)}$ has dimension~$k$.

\section{Decidability}\label{sec:decid}
We formalise the connection between $\orbprob$ and $\simskol$ in the following proposition.
\begin{proposition} \label{prop:reduction}
	%
	%
	%
	%
	%
	For any $k\leq d$, the problems
	$\orbprob_{\KK}(d,d-k)$ and $\simskol_\KK(d,k)$
	are equivalent.
\end{proposition}
\begin{proof}
	Let $(M,\bm{x}, S)$ be an instance of~$\orbprob_\KK(d,d-k)$ and let $S \subseteq \KK^d$ be a subspace spanned by~$\{\bm{v}_1, \dots, \bm{v}_{d-k}\}$. 
	Consider a basis $\{\bm{u}_1, \dots, \bm{u}_k \}$ of the orthogonal complement~$S^\perp$. 
	Then $M^n \bm{x} \in S$ iff $\bm{u}_i^\top M^n \bm{x} =0$ for all $i$. 
	Now, $\bm{u}_1^\top M^n \bm x, \dots , \bm{u}_k^\top M^n \bm{x}$ are all LRS that satisfy the same linear recurrence relation of order~$d$. Indeed, let
	\begin{align*}
		g_M(X) = X^d - c_1X^{d-1} - \dots - c_{d-1}X - c_d
	\end{align*}
	be the characteristic polynomial of $M$, 
	where $c_d \neq 0$ as $M$ is invertible\footnote{See \cref{rmk:mu_jordan_0} and subsequent discussion.}. By the Cayley--Hamilton theorem, $g_M(M) = 0$, and hence, for all~$n$ and $1\leq i\leq k$,
	\begin{align*}
		\bm{u}_i^\top g_M(M)M^n \bm{x} = \bm{u}_i^\top M^{n+d} \bm{x} - \dots - c_d \bm{u}_i^\top M^n \bm{x} = 0\, .
	\end{align*}
	The LRS are linearly independent;
	suppose that $\sum_{i=1}^k \alpha_i \bm{u}_i^\top M^n \bm{x} = 0$ holds for all $n$ for some $(\alpha_1, \dots, \alpha_k) \in \KK^k \setminus \{\bm{0}\}$. Then 
	\begin{align*}
		\left(\sum_{i=1}^k \alpha_i \bm{u}_i \right)^\top M^n \bm{x} = 0
	\end{align*}
	for all $n$. Since the vectors $\bm{u}_i$, $i = 1, \dots, k$, are linearly independent, we have $\sum_{i=1}^k \alpha_i \bm{u}_i \neq 0$ so the orbit $\orbit(M,\bm{x})$ is contained in a proper subspace of $\KK^d$ which contradicts our original assumption. Finally, the characteristic roots of the LRS are exactly the eigenvalues of~$M$, so the non-degeneracy of $M$ implies the non-degeneracy of the LRS.
	
	For the other direction, consider $k$ linearly independent, non-degenerate LRS $u^{(1)}, \dots, u^{(k)}$ that satisfy the same recurrence relation \[
	u_n^{(i)} = c_1 u_{n-1}^{(i)}+ \dots + c_d u_{n-d}^{(i)},
	\]
	where $c_d \neq 0$.
	Take $M$ to be the companion matrix of the characteristic polynomial, so
	$M = \begin{pmatrix}
		\bm{e}_2 \mid \dots \mid \bm{e}_{d} \mid \bm{w}
	\end{pmatrix}$, 
	where $\bm{w} = (c_d, \dots, c_1)^\top$.
	For each $i \in \{1, \dots, k\}$, let $\bm{u}_{\mathrm{init}}^{(i)} \in \KK^d$
	denote the row vector 
	$\bm{u}_{\mathrm{init}}^{(i)} = \left(u^{(i)}_0, \dots, u^{(i)}_{d-1}\right)$ 
	of initial values of the sequence $u^{(i)}_n$.
	These~$k$ vectors $\left\{\bm{u}_{\mathrm{init}}^{(1)}, \dots, \bm{u}_{\mathrm{init}}^{(k)}\right\}$ are linearly independent over~$\KK$.
	Now, the instance of the Simultaneous Skolem Problem is positive if and only if there exists~$n\in\N$ such that
	\begin{align*}
		&n \in \zero(u^{(1)}, \dots, u^{(k)})
		\Leftrightarrow
		u^{(1)}_n = \dots = u^{(k)}_n = 0
		\Leftrightarrow \\
		&\bm{u}_{\mathrm{init}}^{(1)}M^n\bm{e}_1 = \dots = \bm{u}_{\mathrm{init}}^{(k)}M^n\bm{e}_1 = 0.
	\end{align*}
	This is equivalent to
	$M^n\bm{e}_1 \in S = \left\{\bm{u}_{\mathrm{init}}^{(1)}, \dots, \bm{u}_{\mathrm{init}}^{(k)}\right\}^\perp$ for some~$n$, where $S$ is
	a $d-k$-dimensional subspace of~$\KK^d$. 
	Notice that $\dim(M,\bm{e}_1) = d$,
	since $M^{j-1} \bm{e}_1 = \bm{e}_{j}$ for $1 \leq j \leq d$, and $\bm{e}_1, \dots, \bm{e}_d$ span $\KK^d$. 
	Moreover, $(M,\bm e_1)$ has stable dimension as $M$ is invertible, and $M$ is non-degenerate since the $u^{(i)}$ are. Thus $(M,\bm e_1)$ is reduced and $(M,\bm e_1,S)$ forms an instance of $\orbprob(d,d-k)$.
\end{proof}

\subsection{Solving \simskol}
\begin{definition}[\cite{bacik_completing_2025}] \label{def:MSTV_class}
	The MSTV\footnote{Stands for Mignotte, Shorey, Tijdeman, and Vereshchagin.} class consists of algebraic LRS that have at most 3 dominant roots with respect to some Archimedean absolute value or at most 2 dominant roots with respect to some non-Archimedean absolute value.
\end{definition}
The term was coined in~\cite{lipton_skolemPC_2022}
and named after the authors of the seminal papers~\cite{tijdeman_mst1984,vereshchagin_1985}, who proved decidability of the Skolem Problem for all non-degenerate LRS in the MSTV class.
In contrast to the definition in~\cite{lipton_skolemPC_2022},
the LRS we consider are over~$\Alg$;
since there are several embeddings of an arbitrary number field into $\CC$, we have multiple Archimedean absolute values we may use, which will generally have different numbers of dominant roots, while for $\Q$-LRS only the modulus is used.


%
%
The idea we employ to decide $\simskol$ for $k$ LRS $u^{(1)}, \dots, u^{(k)}$ is to take a linear combination
\begin{align*}
	w_n = \sum_{j=1}^{k} \beta_j u_n^{(j)}
\end{align*}
where we choose the $\beta_j \in \Alg$ to eliminate sufficiently many dominant roots and force $w$ to fall into the MSTV class. 
Then all simultaneous zeros of $u^{(1)}, \dots, u^{(k)}$ are zeros of a non-zero non-degenerate $w$, so the Simultaneous Skolem Problem may be solved by finding the finite zero set~$\zero(w)$ 
and checking whether for some $n\in \zero(w)$ we have $u_{n}^{(i)} = 0$ for each $1\leq i \leq k$.

Unfortunately, this method has a complication. Sometimes, the coefficients in front of the dominant roots in the exponential-polynomial representation of the LRS will be particularly pathological, and there is no way to choose a linear combination that leaves at most 3 dominant roots without eliminating all of them (at which point there are ``new'' dominant roots of second-largest absolute value, which complicates application of the MSTV class). Let us illustrate this vexing situation with an example.
\begin{example}
	Choose $\lambda_1, \dots, \lambda_d \in \Alg$ such that $|\lambda_1| = \dots = |\lambda_4| > |\lambda_5| \geq \dots \geq |\lambda_d|$ and 
	consider two LRS $u^{(1)}, u^{(2)}$ given by
	\begin{align*}
		u_n^{(i)} = \alpha_1^{(i)}\lambda_1^n + \alpha_2^{(i)} \lambda_2^n + \alpha_3^{(i)} \lambda_3^n + \alpha_4^{(i)} \lambda_4^n + \sum_{j=5}^d \alpha_j^{(i)} \lambda_j^n
	\end{align*}
	with linearly independent vectors $(\alpha_1^{(i)},\alpha_2^{(i)},\alpha_3^{(i)},\alpha_4^{(i)}, \dots, \alpha_d^{(i)})$ for $i=1,2$. 
	Suppose first that the coefficients of the $|\cdot|$-dominant roots are linearly independent, that is, 
	\[
	\left( \alpha_1^{(1)}, \alpha_2^{(1)}, \alpha_3^{(1)}, \alpha_4^{(1)} \right) \neq t \cdot \left( \alpha_1^{(2)}, \alpha_2^{(2)}, \alpha_3^{(2)}, \alpha_4^{(2)} \right)
	\]
	for all $t \in \overline \Q$. 
	Then it is easy to see that we can choose a linear combination $v_n = \beta_1 u_n^{(1)} + \beta_2 u_n^{(2)}$ that eliminates precisely between 1 and 3 of the dominant roots $\lambda_1, \dots, \lambda_4$, so $ \LRS{v}$ is in the MSTV class. 
	However, if the coefficients of the dominant roots are linearly dependent so that
	\[
	\left( \alpha_1^{(1)}, \alpha_2^{(1)}, \alpha_3^{(1)}, \alpha_4^{(1)} \right) = t \cdot \left( \alpha_1^{(2)}, \alpha_2^{(2)}, \alpha_3^{(2)}, \alpha_4^{(2)}\right)
	\]
	for some $t \in \overline \Q$, then every linear combination $v_n = \beta_1 u_n^{(1)} + \beta_2 u_n^{(2)}$ has exactly 0 or 4 of the dominant roots $\lambda_1 , \dots , \lambda_4$ remaining. 
	In particular, if $|\lambda_5| = \dots = |\lambda_d|$ and $d \geq 8$, it is not possible to take a linear combination that has at most~3 $|\cdot|$-dominant roots. 
\end{example}
We note one saving grace in this situation: if the dominant terms are linearly dependent, we can eliminate more roots than expected. 
In particular, if $d=8$, we can take a linear combination with at most~4 characteristic roots, which does lie in the MSTV class by the argument of \cite{bacik_completing_2025}. 
This is the behaviour we shall exploit in our main results. 
Pathological behaviour aside, it is key to our approach to understand how many $|\cdot|_v$-dominant terms LRS may have, for each absolute value $|\cdot|_v$. This is the content of the following theorem, which is a generalisation of \cite[Lemma 4.3]{bacik_completing_2025}.
\begin{theorem} \label{thm:dominant}
	Let~$\lambda_1, \dots, \lambda_r$ be $r > 1$ algebraic numbers such that no quotient $\lambda_i/\lambda_j$ for any $i \neq j$ is a root of unity and such that
	\begin{align*}
		|\lambda_1| = |\lambda_2| = \dots = |\lambda_r|.
	\end{align*}
	Consider an exponential polynomial 
	\begin{equation}\label{eq:exppoly-f}
		f(\lambda_1, \dots, \lambda_r) = \sum_{i=1}^rP_i(n)\lambda_i^n = 
		\sum_{i=1}^r \sum_{j=0}^{\deg P_i} c_{i,j} n^j\lambda_i^n
	\end{equation} with
	$t$ terms.
	There exists an absolute value $|\cdot|_v$ with respect to which
	at most $\left \lfloor \frac{t}{2} \right \rfloor$ terms of~\eqref{eq:exppoly-f}
	are dominant.
\end{theorem}
\begin{proof}
	Let $\KK$ be a Galois number field containing $\lambda_1,\dots,\lambda_r$, with Galois group $G$. Consider $\mu_i = \frac{\lambda_i}{\lambda_1}$. 
	Note that dividing each $\lambda_i$ by a constant in this way does not change the number of terms that are dominant with respect to any absolute value. 
	Moreover, $f(\lambda_1, \lambda_2, \dots, \lambda_r) = \lambda_1^n \cdot f(1, \mu_2, \dots, \mu_r)$.
	So it suffices to prove that there is a subset of at most~$\left \lfloor \frac{t}{2} \right \rfloor$ of the terms in $f(1, \mu_2, \dots, \mu_r)$ that are dominant with respect to some absolute value. 
	For a term $cn^k\mu^n$ we refer to the algebraic number~$\mu$ as its base. 

	Now we have
	\begin{align*}
		1 = \mu_2 \overline{\mu}_2 = \dots = \mu_r \overline{\mu}_r 
	\end{align*}
	and by applying $\sigma \in G$ to the equation we get
	\begin{align} \label{eqn:sigma_mu}
		1 = \sigma(\mu_2)\sigma(\overline{\mu}_2) = \dots = \sigma(\mu_r)\sigma(\overline{\mu}_r)
	\end{align}
	for all $\sigma \in G$.

	\emph{Case 1: $|\sigma(\mu_i)| \neq 1$ for some $\sigma,i$}; either $|\sigma(\mu_i)| > 1$ or $|\sigma(\overline{\mu}_i)| > 1$. 
	
	Let $|\cdot|_\sigma$ be the absolute value defined by $|x|_\sigma = |\sigma(x)|$ and let $|\cdot|_\tau$ be the absolute value defined by $|x|_\tau = |\sigma(\overline{x})|$.%
	\footnote{The latter is an absolute value as $\tau(x) = \sigma(\overline{x})$ is a composition of Galois automorphisms so is a Galois automorphism itself.} 
	We argue that at most $\lfloor \frac{t}{2} \rfloor$ terms are $|\cdot|_\sigma$-dominant or at most $\lfloor \frac{t}{2} \rfloor$ terms are $|\cdot|_\tau$-dominant.
	
	We discuss the case $|\sigma(\mu_i)| > 1$. Suppose at least $1 + \left \lfloor \frac{t}{2} \right \rfloor$ terms are dominant with respect to $|\cdot|_\sigma$, then
	there are at least $1 + \left \lfloor \frac{t}{2} \right \rfloor$ terms whose base~$\mu_j$ satisfies $|\sigma(\mu_j)| > 1$.
	Therefore, by \eqref{eqn:sigma_mu} there are at least $1 + \left \lfloor \frac{t}{2} \right \rfloor$ terms such that $|\mu_j|_\tau = |\sigma(\overline{\mu}_j)| < 1 = |1|_\tau$ for their bases~$\mu_j$.
	Each of these terms is non-dominant with respect to $|\cdot|_\tau$. So there are at most 
	\begin{align*}
		t - \left(1 + \left \lfloor \frac{t}{2} \right \rfloor \right)  \leq \left \lfloor \frac{t}{2} \right \rfloor
	\end{align*}
	dominant terms with respect to $|\cdot|_\tau$. Note that the case where we instead assume $|\sigma(\overline{\mu}_i)| > 1$ for some~$i$ is completely symmetrical.
	
	\emph{Case 2: $|\sigma(\mu_i)| = 1$ for all $\sigma,i$}.
	If $v_\pp(\mu_2) = 0$ for all prime ideals $\pp \subseteq \O_\KK$, then $\mu_2 \in \O_\KK$; 
	in fact, $\mu_2$ is a unit of that ring.
	All Galois conjugates of~$\mu_2$ lie in the unit circle, so by Kronecker's theorem $\mu_2$ is a root of unity. This contradicts our assumption that $\lambda_2/\lambda_1$ is not a root of unity.
	Therefore, we have $v_\pp(\mu_2) \neq 0$ for some prime ideal $\pp \subseteq \O_\KK$. Then by the exact same argument as in Case 1, at most $\left\lfloor \frac{t}{2} \right \rfloor$ of the terms 
	are dominant with respect to either $|\cdot|_\pp$ or the non-Archimedean absolute value $|\cdot|_v$ defined by $|x|_v = |\overline x|_\pp$.
\end{proof}
We give an example to show that this bound is optimal---one cannot always guarantee having fewer dominant terms. 
In fact, this bound is optimal even when one assumes extra structure from the $\lambda_i$ arising as characteristic roots of a simple \emph{integer}-valued LRS.
\begin{example} \label{ex:gaussian}
	We work with the ring of Gaussian integers $\Z[i]$. Recall that $\Z[i]$ is the ring of algebraic integers of $\Q(i)$ and is a PID. The Gaussian primes are exactly those of the form $a+bi$ with $a^2 + b^2$ equal to an integer prime, or $a = 0$ and $|b| \equiv 3 \bmod 4$, or $b = 0$ and $|a| \equiv 3 \bmod 4$. 
	
	
	We choose~$s \geq 3$ Gaussian primes $p_1, \dots, p_s \in \Z[i]$ such that there are no pairs of associates\footnote{Recall that~$a, b\in R$ are associates if $a = bu$ for some invertible~$u \in R$.} in the set $\{p_1, \dots,p_s,\overline p_1, \dots, \overline p_s\}$.
	This can be achieved by choosing Gaussian integers of the form $a+bi$ with $|a|,|b| > 1$, as this implies $p_i, \overline p_i$ are not associates. 
	Moreover, to guarantee that $p_i$ and $p_j$, as well as $p_i$ and $\overline{p}_j$, are not associates when $i \neq j$, we select $p_1, \dots, p_s$ so that 
	$|p_i| > |p_{i+1}|$ for all $1 \leq i \leq s-1$.
	
	We consider $\lambda_1, \dots, \lambda_s \in \Z[i]$ 
	given by
	\[
	\lambda_i = \overline{p}_i \prod_{j \neq i} p_j \, .
	\]
	Observe that $\lambda_1, \dots, \lambda_s, \overline{\lambda}_1, \dots, \overline{\lambda}_s$ are pairwise distinct 
	and no two numbers in the set~$\Lambda_{2s}\coloneq\{\lambda_1, \dots, \lambda_s, \overline{\lambda}_1, \dots, \overline{\lambda}_s\}$ are associates, as they have distinct prime factorisations. 
	
	Clearly we have $|\lambda_1| = |\overline{\lambda}_1| = \dots = |\lambda_s| = |\overline{\lambda}_s|$.
	Moreover, they all have equal absolute value with respect to every Archimedean absolute value as the only other Archimedean absolute value is given by $|x|_v = |\overline x|$. 
	Now, $|\cdot|_{p_i},|\cdot|_{\overline{p}_i}$ for $i = 1 , \dots, s$ are the only non-Archimedean absolute values for which not all elements of~$\Lambda_{2s}$ are dominant.
	It is easy to see that for each~$i$ exactly~$s$ elements of~$\Lambda_{2s}$ are dominant in $|\cdot|_{p_i}$, and similarly for~$|\cdot|_{\overline{p}_i}$.
	
	Now, to extend our example to sets of algebraic numbers of cardinality~$r = 2s+1 \geq 7$, let 
	\[
	\lambda_{s+1} = \prod_{i} p_i\overline{p}_i \in \Z \, .
	\]
	Then all numbers in the set $\Lambda_{2s+1} = \{\lambda_1^2 , \smash{\overline{\lambda}}_1^2, \dots ,\lambda_s^2, \smash{\overline{\lambda}}_s^2,\lambda_{s+1}\}$ are dominant with respect to every Archimedean absolute value. For a non-Archimedean absolute value $|\cdot|_{p_i}$ (resp.\ $|\cdot|_{\overline{p}_i}$), exactly $s$ of the numbers are dominant, and all are dominant with respect to every other non-Archimedean absolute value. 
	
	For every $r \geq 6$, we have constructed a set of roots~$\Lambda_r$ such that with respect to any absolute value at least $\lfloor \frac{r}{2} \rfloor$ of the roots are dominant. The cases $r=4,5$ may be handled by taking $\Lambda_4 \coloneq \{\lambda_1, \lambda_2, \smash{\overline{\lambda}}_1, \smash{\overline{\lambda}}_2\} \subseteq \Lambda_6$, and $\Lambda_5 \coloneq \{\lambda_1^2, \lambda_2^2, \smash{\overline{\lambda}}_1^2, \smash{\overline{\lambda}}_2^2,\lambda_4\} \subset \Lambda_7$. These have the desired property that all roots are dominant with respect to all Archimedean absolute values, and for each non-Archimedean absolute value there are at least~2 dominant roots.
	
	All constructed sets $\Lambda_r$, $r \geq 4$,
	arise as sets of characteristic roots of simple integer-valued LRS
	because each $\Lambda_r \subset \Z[i]$ is closed under complex conjugates.
\end{example}

\subsection{MSTV-reducibility}
We now show how to use \cref{thm:dominant} to decide cases of $\simskol$ (and hence of $\orbprob$). 
Let $u_n^{(1)}, \dots, u_n^{(k)}$ be non-degenerate linearly independent LRS satisfying the same recurrence relation. 
Recall that if we can find a linear combination~$w_n$ that lies in the MSTV class, 
then we can compute the zero set~$\zero(u^{(1)}, \dots, u^{(k)})$.

\begin{definition} \label{def:MSTV-reducible}
	An instance $u_n^{(1)},\dots,u_n^{(k)}$ of $\simskol(d,k)$ is \emph{MSTV-reducible} if there exists a 
	non-zero linear combination 
	$w_n = \sum_{i=1}^k \beta_i u_n^{(i)}$ with $\beta_i \in \Alg$ which lies in the MSTV class. 
	Similarly, an instance of $\orbprob(d,t)$ is MSTV-reducible if the instance of $\simskol(d,d-t)$ to which it reduces is MSTV-reducible.
	
	We refer to problems $\orbprob(d,t)$ and $\simskol(d,k)$ as MSTV-reducible (hence decidable) if all their instances are MSTV-reducible.
	
\end{definition}


Recall the representation of the LRS as vectors $\bm u^{(1)}, \dots, \bm u^{(k)}$ of coefficients; let $\pi_v, \pi_{\neg v}$ be respectively the projections onto the dominant and non-dominant components with respect to the absolute value $|\cdot|_v$.
We now discuss the images of the associated subspace~$C \subseteq \Alg^d$ under projections. 
Notice that an instance of~$\simskol(d,k)$ is MSTV-reducible 
if there exists an absolute value $|\cdot|_v$ and a vector $\bm{w} \in C$ such that the projection $\pi_v(\bm{w})$ is non-zero 
and has support of size at most~3 (for Archimedean $|\cdot|_v$) or at most~2 (for non-Archimedean $|\cdot|_v$).\footnote{Note that this actually corresponds to $\bm w$ having at most 3 dominant terms with respect to an Archimedean absolute value or at most 2 dominant terms with respect to a non-Archimedean absolute value, which is slightly stronger than Definition \ref{def:MSTV_class}.}

The following is the lemma that guides our 
search for linear combinations that have few dominant terms.
\begin{lemma}\label{lem:echelon}
	Given an instance $\{u_n^{(i)}\}_{i=1}^k$ of $\simskol(d,k)$, 
	let~$C$ be the subspace spanned by the~$k$ associated vectors of these sequences.
	Let $\ell$ be the number of vector components that correspond to dominant terms with respect to a fixed absolute value~$|\cdot|_v$. Then, for $r:= \dim \pi_v (C)$,
	\begin{enumerate}
		\item there exist $k - r$ linearly independent vectors $\bm{w}_1, \dots, \bm{w}_{k-r} \in C$ with $\pi_v(\bm{w}_1) = \dots = \pi_v(\bm{w}_{k-r}) = \bm{0}$;
		
		\item there exists $\bm{w} \in C$ for which $\pi_v(\bm{w})$ has the support of size at least~1 and at most~$\ell - r + 1$.
	\end{enumerate}
	%
	%
	%
\end{lemma}
\begin{proof}
	We present a constructive proof whence an algorithm computing $\bm{w}, \bm{w}_1, \dots, \bm{w}_{k-r}$ can be deduced.
	Let $\begin{pmatrix} A \mid B \end{pmatrix}$ be the matrix whose $i$-th row comprises the coefficients of the sequence~$u_n^{(i)}$.
	The columns of $A \in \Alg^{k \times \ell}$ correspond to the $|\cdot|_v$-dominant terms, 
	while the columns of $B \in \Alg^{k \times (d - \ell)}$ are the coefficients of the remaining terms.
	We perform elementary row operations on $\begin{pmatrix} A \mid B \end{pmatrix}$ in order to put $A$ into row echelon form.
	Observe that the rows of $\begin{pmatrix} A \mid B \end{pmatrix}$ are linearly independent,
	and so are the rows of the transformed $k \times d$ matrix at any step of this procedure.
	The subspace~$\pi_v(C)$ is spanned by the rows of~$A$, and hence the rank of~$A$ is~$r$.
	
	The rest follows from two observations. First, the row echelon form of a $k \times \ell$ matrix with rank~$r > 0$ has $k - r$ zero rows. 
	Second, the last non-zero row has at least $r-1$ zeros.
\end{proof}

\begin{lemma}\label{lem:growambient}
	If $\ \simskol(d,k)$ is MSTV-reducible, then 
	so is $\simskol(d+1,k+1)$.
\end{lemma}
\begin{proof}
	Let $\{u^{(i)}\}_{i=1}^{k+1}$ be an instance of~$\simskol(d+1,k+1)$.
	Denote the coefficient vector of the $i$-th sequence by $\bm{u}^{(i)}$ and let~$C := \spn(\bm{u}^{(1)}, \dots, \bm{u}^{(k+1)}) \subseteq \Alg^{d+1}$.	
	
	First, we restrict our attention to those vector components that correspond to the characteristic root~$\lambda_1$. 
	We select the non-zero component of~$n^j\lambda_1^n$
	with the greatest~$j$ 
	across all $\bm{u}^{(1)}, \dots, \bm{u}^{(k+1)}$.
	Relabel the vector components so that the selected component is the first one.
	Without loss of generality, the (new) first component of~$\bm{u}^{(1)}$ is non-zero.
	By performing one round of Gaussian elimination, it is easy to see 
	that $C$ contains~$k$ linearly independent vectors~$\bm{v}^{(1)}, \dots, \bm{v}^{(k)}$ such that the first component of each vector is zero.

	Now consider the vectors $\bm{w}^{(1)}, \dots, \bm{w}^{(k)} \in \Alg^d$ obtained from $\bm{v}^{(1)}, \dots, \bm{v}^{(k)}$ by eliminating the first component. 
	By assumption there exists a linear combination $w_n = \sum_{i=1}^k \beta_i w_n^{(i)}$ of their associated sequences, an LRS of order at most~$d$, which lies in the MSTV class.
	However, then the linear combination $v_n := \sum_{i=1}^k \beta_i v_n^{(i)}=w_n$ of the associated sequences of $\bm v^{(1)}, \dots, \bm v^{(k)}$ -- a linear combination of the LRS $u^{(1)}, \dots, u^{(k+1)}$ -- lies in the MSTV class.	
\end{proof}
The corollary below follows from applying~\cref{lem:growambient} to the fact that~$\simskol(4,1) = \skolem(4)$ is MSTV-reducible~\cite{bacik_completing_2025}.
\begin{corollary}\label{cor:manyseq}
	$\simskol(d,k)$ is MSTV-reducible for any $d, k$ such that $0 \leq d-k \leq 3$.
\end{corollary}
	%

By~\cref{cor:manyseq} we moreover have that $\orbprob(d,d-k)$ is decidable for $d-k \leq 3$ and any~$d$. 
We may also interpret~\cref{lem:growambient} in terms of the reduced Subspace Orbit Problem: if~$\orbprob(d,t)$ is MSTV-reducible,
then $\orbprob(d+1,t)$ is MSTV-reducible as well. 
From now on, our goal is to find the decidable $\simskol(d,k)$ with the least~$d$
for each fixed~$d-k$. 

Before moving to the decidability results for~$\simskol$, we note a useful corollary of~\cref{thm:dominant} which we will use repeatedly.
\begin{corollary} \label{cor:dom_roots}
	For any instance of $\simskol(d,k)$ with characteristic roots $\lambda_1, \dots, \lambda_s$, and for any $0 < t \leq d$ such that $d < 2t+2 - \left\lfloor \frac{t+1}{2} \right\rfloor$ either there are at most $t$ dominant terms with respect to some Archimedean absolute value, or there are at most $\left\lfloor \frac{t+1}{2} \right\rfloor + d -(t+1)$ dominant terms with respect to some non-Archimedean absolute value. 
\end{corollary}
\begin{proof}
	Suppose at least $t+1$ terms are dominant with respect to every Archimedean absolute value. 
	Let~$r$ be the number of distinct characteristic roots of the largest modulus, i.e.\
	\begin{align*}
		|\lambda_1| = \dots = |\lambda_r| > |\lambda_{r+1}| \geq \dots \geq |\lambda_s|.
	\end{align*}
	
	Then there are at least~$t+1$ terms with bases $\lambda_1, \dots, \lambda_r$. We denote the number of such terms by~$T$.
	By~\cref{thm:dominant} there is an absolute value $|\cdot|_v$ with respect to which at most $\left \lfloor \frac{T}{2} \right \rfloor$ terms with bases among $\lambda_1 , \dots, \lambda_r$ are dominant for~$|\cdot|_v$. 
	Therefore, the number of all dominant terms
	is bounded from above by $\left\lfloor \frac{T}{2} \right\rfloor + d - T$.
	This expression attains its maximum (provided $T \geq t+1$) at $T = t+1$, 
	thus there are at most $\left \lfloor \frac{t+1}{2} \right \rfloor + d - (t+1)$ terms dominant with respect to~$|\cdot|_v$.
	
	The assumption $d< 2t+2 - \left\lfloor \frac{t+1}{2} \right\rfloor$ is equivalent to 
	$\left\lfloor \frac{t+1}{2} \right\rfloor + d - (t+1) < t+1$, and so there are at most~$t$ dominant terms with respect to $|\cdot|_v$.
	We conclude from this that $|\cdot|_v$ is non-Archimedean.
\end{proof}

\begin{theorem}\label{thm:62}
	$\simskol(6,2)$ is MSTV-reducible.
\end{theorem}
\begin{proof}
	Let $\{u^{(1)},u^{(2)}\}$ be an instance of $\simskol(6,2)$. 
	From~\cref{cor:dom_roots} there is an absolute value $|\cdot|_v$ such that one of the following holds:
	\begin{enumerate}
		\item $|\cdot|_v$ is Archimedean and at most 4 terms are $|\cdot|_v$-dominant,
		\item $|\cdot|_v$ is non-Archimedean and at most 3 terms are $|\cdot|_v$-dominant.
	\end{enumerate}
	In fact, we may assume both $u^{(1)},u^{(2)}$ have exactly 4 dominant terms if~(1) is the case, and exactly 3 dominant terms if~(2) is the case; otherwise one of $u^{(1)},u^{(2)}$ would be in the MSTV class, which immediately yields MSTV-reducibility of the instance.
	
	\textbf{Case (1)}: $|\cdot|_v$ is Archimedean with exactly~4 dominant terms in $u^{(1)},u^{(2)}$. 
	
	The associated vectors $\bm{u}^{(1)}$ and $\bm{u}^{(2)}$ span a subspace~$C\subset \Alg^6$ where $\dim C = 2$.
	By~\cref{lem:echelon} there exists $\bm{w}\in C$ such that either $\pi_v(\bm{w}) = \bm{0}$ or $1 \leq \#\supp(\pi_v(\bm{w})) \leq 3$.
	If $\pi_v(\bm{w}) = \bm{0}$, the sequence~$\LRS{w}$ associated with~$\bm{w}$ has order at most~2.
	If $1 \leq \#\supp(\pi_v(\bm{w})) \leq 3$, 
	then $\LRS{w}$ has at most $3$ dominant terms with respect to the Archimedean $|\cdot|_v$.
	In both cases $\LRS{w}$ is in the MSTV class. 
	
	\textbf{Case (2)}: $|\cdot|_v$ is non-Archimedean with exactly~3 dominant terms in $u^{(1)},u^{(2)}$.
	
	Applying~\cref{lem:echelon} again, we find
	$\bm{w}\in C$ such that either $\pi_v(\bm{w}) = \bm{0}$ or
	$1 \leq \#\supp(\pi_v(\bm{w})) \leq 2$.
	In the former case, the sequence~$\LRS{w}$ associated with~$\bm{w}$ has order at most~3.
	In the latter case, $\LRS{w}$ has at most $2$ dominant terms with respect to the non-Archimedean~$|\cdot|_v$.
	In both cases $\LRS{w}$ is in the MSTV class. 
\end{proof}
\begin{theorem}\label{thm:94}
	$\simskol(9,4)$ is MSTV-reducible.
\end{theorem}
\begin{proof}
	Let $\left\{\LRS{u}^{(1)},\LRS{u}^{(2)},\LRS{u}^{(3)},\LRS{u}^{(4)} \right\}$ be an instance of $\simskol(9,4)$. 
	From~\cref{cor:dom_roots} there exists either an Archimedean absolute value with at most~6 dominant terms or a non-Archimedean absolute value with at most~5 dominant terms appearing in the LRS $u^{(1)}, \dots, u^{(4)}$. 
	
	\textbf{Case (1)}: Assume first that $|\cdot|_v$ is an Archimedean absolute value with $\ell \leq 6$ dominant terms.
	The associated vectors $\bm{u}^{(i)}$, $1 \leq i \leq 4$, span a four-dimensional subspace~$C$ of $\Alg^9$,
	and its projection $\pi_v(C)$ has dimension~$r$.
	By~\cref{lem:echelon},
	if $\ell - r + 1 \leq 3$, then there exists $\bm{w}\in C$ such that $1 \leq \#\supp(\pi_v(\bm{w})) \leq 3$.
	Then the proof is concluded as~$\LRS{w}$ lies in the MSTV class.
	We discuss the situation when $\ell - r + 1 \geq 4$. 
	Since $\ell \leq 6$, we have $r \leq 3$.
	Note also that we only consider $\ell \geq 4$, as $\LRS{u}^{(1)}$ is in the MSTV class otherwise.
	Recall also from~\cref{lem:echelon} that there exist $4-r$ linearly independent vectors~$\bm{w}_1, \dots, \bm{w}_{4-r}$ such that $\bm{w}_i\in C$ with $\pi_v(\bm{w}_i) = \bm{0}$.
	The problem thus reduces to an instance of a $\simskol(9-\ell, 4-r)$.
	By carefully analysing the possible values of $\ell$ and $r$,
	it is easy to see that all these problems are in the MSTV class.
	
	
	\textbf{Case (2)}: Now we can assume that there exists a non-Archimedean absolute value~$|\cdot|_v$ with $\ell \leq 5$ dominant terms.
	Then, similarly to the Archimedean case,
	we first note that $\ell - r +1 \leq 2$ implies the existence of~$\bm{w} \in C$ with $1 \leq \#\supp(\pi_v(\bm{w})) \leq 2$ whose associated sequence~$\LRS{w}$ lies in the MSTV class.
	
	We thus assume $\ell - r + 1 \geq 3$, as well as $\ell \geq 3$.
	As before, we deduce $r \leq 3$ and again by \cref{lem:echelon} the problem can always be reduced to a $\simskol(9-\ell, 4-r)$ instance. 
	All these instances are MSTV-reducible since $9-\ell - (4-r) = 6 - (\ell - r + 1) \leq 3$ and we can apply~\cref{cor:manyseq}.
\end{proof}
\begin{theorem}
	$\simskol(12,6)$ is MSTV-reducible.
\end{theorem}
\begin{proof}
	Let $\left\{\LRS{u}^{(1)},\LRS{u}^{(2)},\ldots,\LRS{u}^{(6)} \right\}$ be an instance of $\simskol(12,6)$ with characteristic roots $\lambda_1, \dots, \lambda_s$, where $s \leq 12$.
	
	\textbf{Case (1)}: Assume first that $|\cdot|_v$ is an Archimedean absolute value with $\ell \leq 8$ dominant terms appearing in $u^{(1)}, \dots, u^{(6)}$. 
	We denote by~$\pi_v(C)$ the projection to the $|\cdot|_v$-dominant terms of the subspace~$C \subseteq \Alg^{12}$ spanned by the associated vectors of~6 sequences. It has dimension~$r \leq \min\{\ell, 6\}$.
	Similar to the proofs of~\cref{thm:62,thm:94},
	the existence of a linear combination of $\left\{\LRS{u}^{(i)} : 1 \leq i \leq 6\right\}$ that lies in the MSTV class
	is guaranteed in each of the following cases:
	\begin{itemize}
		\item $\ell-r+1 \leq 3$: then, by~\cref{lem:echelon}, there exists $\bm{w} \in C$ such that $1 \leq \#\supp(\pi_v(\bm{w})) \leq 3$. This implies that the associated sequence is in the MSTV class.
		\item $\ell-r+1 \geq 4$: now, by~\cref{lem:echelon} we can find $6-r$ linearly independent vectors~$\bm{w}_1, \dots, \bm{w}_{6-r}$ such that $\bm{w}_i\in C$ with $\pi_v(\bm{w}_i) = \bm{0}$.
		By considering $\bm w_1, \dots, \bm w_{6-r}$, the problem reduces to an instance of $\simskol(12-\ell, 6-r)$.
		Now, by our assumption, $(12 - \ell) - (6 -r) = 6 - \ell + r = -(\ell - r + 1) + 7 \leq 3$.
		Therefore, each of these~$\simskol(12-\ell, 6-r)$ instances is MSTV-reducible by~\cref{cor:manyseq}.
	\end{itemize}
	
	\textbf{Case (2a)}: If $\ell \leq 6$ terms are dominant with respect to a non-Archimedean absolute value~$|\cdot|_v$, then we can proceed similarly:
	\begin{itemize}
		\item $\ell-r+1 \leq 2$: then, by~\cref{lem:echelon}, there exists $\bm{w} \in C$ such that $1 \leq \#\supp(\pi_v(\bm{w})) \leq 2$. This implies that the associated sequence is in the MSTV class.
		\item $\ell-r+1 \geq 3$: by~\cref{lem:echelon} we can find $6-r$ linearly independent vectors~$\bm{w}_1, \dots, \bm{w}_{6-r}$ such that $\bm{w}_i\in C$ with $\pi_v(\bm{w}_i) = \bm{0}$.
		By eliminating the $|\cdot|_v$-dominant terms, the problem reduces to an instance of a $\simskol(12-\ell, 6-r)$.
		Moreover, $(12 - \ell) - (6 -r) = 6 - \ell + r = -(\ell - r + 1) + 7 \leq 4$.
		The new case is when $12 - \ell - (6-r) =4$. Due to $12 - \ell \geq 6$, these instances are MSTV-reducible from~\cref{thm:62} and \cref{lem:growambient}.
		All other $\simskol(12-\ell, 6-r)$ instances we obtain are MSTV-reducible by~\cref{cor:manyseq}.
	\end{itemize}
	
	However, we now come to the first case that requires further analysis. 
	Assume now that there are at least 9 dominant terms for every Archimedean absolute value. 
	By~\cref{cor:dom_roots} (setting $t:=8$) we are guaranteed that there are at most~7 dominant terms with respect to some non-Archimedean absolute value.
	
	\textbf{Case (2b)}: We have already considered the case of at most~6 dominant terms;
	hence we proceed under the assumption that $|\cdot|_\pp$ is the non-Archimedean absolute value with exactly~7 dominant terms constructed as in~\cref{cor:dom_roots}.
	
	Let~$m$ denote the number of distinct characteristic roots with the largest modulus, without loss of generality,
	\[
	|\lambda_1| = \dots = |\lambda_m| > |\lambda_{m+1}| \geq \dots \geq |\lambda_s| \, .
	\]
	Let~$T \in \{9,10,11,12\}$ be the number of terms associated with~$\lambda_1, \dots, \lambda_m$.
	Since there are exactly~7 dominant terms with respect to~$|\cdot|_\pp$,
	and since there are at most $\lfloor \frac{T}{2}\rfloor$ of them associated with~$\lambda_1, \dots, \lambda_m$\footnote{Recall that in \cref{cor:dom_roots}, the dominant terms are constructed by applying \cref{thm:dominant} to $\lambda_1, \dots, \lambda_m$.}, it is easy to see that only two cases are possible.
	Either $T = 9$ and exactly~3 terms with bases among $\lambda_{m+1}, \dots, \lambda_s$ are $|\cdot|_\pp$-dominant. 
	Or $T = 10$ and exactly~2 terms with bases among $\lambda_{m+1}, \dots, \lambda_s$ are $|\cdot|_\pp$-dominant. 
	Crucially, in both cases all terms with bases among $\lambda_{m+1}, \dots, \lambda_s$ are $|\cdot|_\pp$-dominant.
	
	We now employ~\cref{lem:echelon} for~$|\cdot|_\pp$.
	Let the projection $\pi_\pp(C)$ to the $|\cdot|_\pp$-dominant terms have dimension~$r$.
	If $r=6$, then we have $7 - r + 1 = 2$, and so there exists $\bm{w}\in C$ such that $1 \leq \#\supp(\pi_\pp(\bm{w})) \leq 2$.
	Then, the associated sequence is in the MSTV class.
	Otherwise, if $r \leq 5$, we can find $\bm{w}\in C$ such that $ \#\supp(\pi_\pp(\bm{w})) = 0$.
	Recall from the previous paragraph that $\supp(\pi_\pp(\bm{w}))$ 
	contains no components that correspond to terms with bases $\lambda_{m+1}, \dots, \lambda_s$.
	In other words, the associated sequence of~$\bm{w}$ has at most~5 non-zero terms, all of them having bases of the same modulus~$|\lambda_1|$. 
	By~\cref{thm:dominant}, there exists an absolute value for which at most~2 of these terms are dominant,
	placing this sequence in the MSTV class.
\end{proof}
These theorems show that $\orbprob(d,t)$ is decidable for $(d,t) = (6,4)$, $(9,5)$, and $(12,6)$. 
Furthermore, those are the smallest~$d$ for which $\orbprob(d,t)$ with, respectively, $t=4,5,6$ can be decided using the properties of the MSTV class.

By~\cref{cor:dom_roots}, the inequality \begin{equation} \label{eq:when-mstv}
	d \leq 2k + 4 - \left \lfloor \frac{k+3}{2} \right \rfloor
\end{equation} implies 
that either there exists an Archimedean absolute value~$|\cdot|_v$ with at most~$k+2$ dominant terms, 
or there is a prime ideal $\pp \subseteq \O_\KK$ such that
the absolute value~$|\cdot|_\pp$ has at most~$k+1$ dominant terms.

The condition~\eqref{eq:when-mstv} is optimal, as witnessed by the following family of examples. 
In particular, it is easy to see that $\simskol$ with (5,1), (8,3) and (11,5) are not MSTV-reducible.
\begin{example}\label{ex:notmstv}
	Given integers $d,k$ such that \begin{equation}\label{eq:when-not-mstv}
		d > 2k + 4 - \left \lfloor \frac{k+3}{2} \right \rfloor,
	\end{equation} 
	we show instances of $\simskol(d,k)$ that are not MSTV-reducible. 
	
	First, we consider odd~$k$; that is, the case where $k+3$ is even. 
	Let~$s \coloneqq (k+3)/2$ and notice that $d - s > k+1$ follows. 
	Let~$\Lambda_{2s}$ be as defined in~\cref{ex:gaussian}.
	Let $\lambda_{s+1}, \dots, \lambda_{d-s}$ be Gaussian primes\footnote{As before, chosen so that they are not associates to each other or their conjugates.} smaller in modulus than any prime factor of~$\lambda_i$ with $1 \leq i \leq s$.
	We choose $d$ characteristic roots
	$\Lambda_{2s} \cup \{\lambda_{s+1}, \dots, \lambda_{d-s}\}$. 
	
	Using the argument of~\cref{ex:gaussian}, we see that
	for every Archimedean absolute value
	the dominant roots are $\lambda_1, \dots, \lambda_s, \overline{\lambda}_1, \dots, \overline{\lambda}_s$,
	that is, there are~$2s = k+3$ dominant roots. 
	Among non-Archimedean absolute values,
	we identify the absolute values $|\cdot|_p$ 
	where $p$ is a prime factor of some element in~$\Lambda_{2s}$. 
	For each of them, there are 
	$s + d-2s \geq k+2$ dominant roots.
	Moreover, for the absolute values $|\cdot|_{\lambda_i}$  with $i \geq s+1$, there are $d-1$ dominant terms, and for all other non-Archimedean absolute values all~$d$ roots are dominant.
	
	Now we choose our $k$ LRS as $u_n^{(1)}, \dots, u_n^{(k)}$ with
	\begin{align*}
		u_n^{(i)} = \sum_{j=1}^{s} j^{i-1} \lambda_j^n + \sum_{j=1}^{s} (j+s)^{i-1} \overline{\lambda}_j^n + \sum_{j=s+1}^{d-s} (j+s)^{i-1} \lambda_j^n \, .
	\end{align*}
	Let $|\cdot|_v$ be an arbitrary absolute value and let $\ell$ denote the number of dominant roots with respect to it. 
	Denote by $C \subseteq \Alg^d$ the vector subspace spanned by the associated vectors of $u^{(1)},\dots, u^{(k)}$.
	We show that~$\#\supp(\pi_v(\bm{w})) \geq \ell-(k-1)$ for any non-zero vector~$\bm{w} \in C$. 
	Indeed, assume that among components that correspond to $|\cdot|_v$-dominant roots there are $k$ components $I = \{i_1, \dots, i_k\}$ such that for some~$\bm{w}\neq\bm{0}$ we have $\bm{w}_i = 0$ for all~$i\in I$.
	Then $\bm{u}^{(1)}, \dots, \bm{u}^{(k)}$ restricted to components in~$I$ are linearly dependent; yet these vectors constitute a $k \times k$ Vandermonde matrix $V(i_1, \dots, i_k)$ whose determinant is non-zero, which is a contradiction.
	
	Consequently, there is no linear combination of the LRS with at most 3 Archimedean dominant roots or at most 2 non-Archimedean dominant roots, so this $(d,k)$-instance is not MSTV-reducible. 
	
	We now prove that some $\simskol(d,k)$ instances that satisfy~\eqref{eq:when-not-mstv} are not MSTV-reducible, 
	\emph{for all even~$k$ as well}.
	Assume that for some even~$k$ there exists~$d$ as in~\eqref{eq:when-not-mstv} such that $\simskol(d,k)$ is MSTV-reducible. 
	From~\cref{eq:when-not-mstv} we have \[d+1 > 2k+4 -  \left\lfloor \frac{k+3}{2}\right\rfloor + 1 = 2(k+1)+4 - \left\lfloor \frac{(k+1)+3}{2}\right\rfloor\, , \]
	and therefore, $\simskol(d+1,k+1)$ is MSTV-reducible by~\cref{lem:growambient}. 
	This contradicts what we have just proved about the odd~$k$ case.
\end{example}



We will now show that for every parameter~$d-k$ there exists sufficiently large~$d$ such that $\simskol(d,k)$ is MSTV-reducible.
This is most naturally formulated in terms of the reduced Subspace Orbit Problem:
for every target dimension~$t$ there exists~$d$ such that 
$\orbprob(d,t)$ is MSTV-reducible, and hence is decidable.

%
\mainresult*
\begin{proof}
	We have shown the statement for all $t \leq 6$. 
	Suppose that $t \geq 7$ and that the theorem has been proven for every target dimension smaller than~$t$. Pick $d,k$ such that $d-k = t$ and $d-k \leq 2 \log_3 d$, and consider an arbitrary instance of $\simskol(d,k)$.
	Notice that for $d-k\geq7$ we have $d-k < 3^{\frac{d-k}{2} -1}$, and so
	$\frac{d}{3} \geq 3^{\frac{d-k}{2}-1} > d-k$. 
	This implies $k > \frac{2d}{3}$.
	
	We apply~\cref{cor:dom_roots} with a parameter~$\frac{2d}{3}$, from which we argue that there exists an absolute value~$|\cdot|_v$ with $\ell \leq \frac{2d}{3}$ dominant terms, either Archimedean or not.
	We can now apply~\cref{lem:echelon} for the subspace~$C\subseteq\Alg^d$ spanned by $k$ vectors of our sequences.
	Let $r:= \dim \pi_v(C)$. Notice that $r \leq \min\{\ell, k\} = \ell$. 
	If $\ell - r \in \{0,1\}$, then by~\cref{lem:echelon} 
	there exists a sequence with at most~2 $|\cdot|_v$-dominant terms. Therefore, the problem is MSTV-reducible.
	
	
	We now assume that $\ell - r \geq 2$. Since $k - r \geq 1$, we deduce from~\cref{lem:echelon}
	that this instance reduces to $\simskol(d-\ell,k-r)$ by eliminating the $|\cdot|_v$-dominant terms. 
	Notice that $d-\ell - (k-r) \leq d-k-2$ by our latest assumption.
	Moreover, \[\ell \leq \frac{2d}{3} \Leftrightarrow d-\ell \geq \frac{d}{3} \Rightarrow d - \ell \geq 3^{\frac{d-k}{2}-1} = 3^{\frac{d-k-2}{2}} ,\]
	so we have $d-\ell \geq 3^{\frac{d-\ell - (k-r)}{2}}$.
	The problem we reduced to has a strictly smaller target dimension, as $d-\ell - (k-r) < d -k$ follows from $\ell - r \geq 2$.
	Consequently, $\simskol(d-\ell, k-r)$ is MSTV-reducible by induction on target dimension.
\end{proof}

\section{Hardness} \label{sec:hardness}
In this section we show that if for some~$C\in (0,1)$ there is an algorithm to decide $\orbprob(d,t)$ for all $t,d$ such that $t \leq Cd$, 
then the Skolem Problem is decidable.
%
%
%

Given a $\Alg$-LRS~$u$, pick multiplicatively independent algebraic numbers $\mu_1, \dots, \mu_s \neq 0$ such that we may write
\begin{align*}
	u_n = L(\mu_1^n, \dots, \mu_s^n)\cdot P(n+1,\mu_1^n, \dots, \mu_s^n)
\end{align*}
for some Laurent monomial $L(x_1, \dots, x_s)$ and some polynomial $P(x_0, \dots, x_s)$ with algebraic coefficients (let $x_0$ correspond to $n+1$ and $x_i$ correspond to $\mu_i^n$).\footnote{The $\mu_i$ are not necessarily the characteristic roots of $u$ because one may need to take roots. For example if there is a multiplicative relation $\lambda_1^3 = \lambda_2^2$, one can take $\mu_1 = \lambda_1^{1/2}$, for some definition of $\lambda_1^{1/2}$.} 
Note that monomials of $P$ correspond to terms in the exponential polynomial representation of $u$. Note also that $u_n = 0$ if and only if $P(n+1,\mu_1^n, \dots, \mu_s^n) = 0$. We want to find polynomials $Q_1, \dots, Q_k$ such that the LRS defined by $Q_1P, \dots, Q_kP$ form an instance of~$\simskol(d,k)$ with $d$ small compared to $k$. 

\begin{lemma}\label{ssp-family}
	Given a non-degenerate LRS~$u$ with associated polynomial $P(x_0, \dots, x_s)$,
	there exists a family of $\simskol$ instances $I_1, I_2, \dots$
	such that each $I_\ell$ consists of $k = \ell^{s+1}$ non-degenerate LRS that satisfy a recurrence relation of order $d \leq \ell^{s+1}+B\ell^s$
	for a constant~$B$ that only depends on~$P$.
	Moreover, the set of simultaneous zeros $\zero(I_\ell)$ is equal to $\zero(u)$.
\end{lemma}
\begin{proof}
	We will identify monomials $x_0^{m_0} x_1^{m_1} \dots x_s^{m_s}$ with points in $\Z^{s+1}$ by
	\begin{align*}
		x_0^{m_0} x_1^{m_1} \dots x_s^{m_s} \mapsto (m_0, m_1, \dots, m_s) \, .
	\end{align*}
	Let $R$ be the set of points $(m_0, m_1, \dots, m_s)$ such that $P$ has a term $x_0^{\nu} x_1^{m_1} \dots x_s^{m_s}$ for $0 \leq m_0 \leq \nu$.\footnote{This different treatment of $x_0$ is necessary due to the difference in how powers of $n$ affect the order of an LRS.} 
	
	We fix an arbitrary integer~$\ell$ and construct the instance~$I_\ell$.
	Let \[C_\ell = \{(m_0,m_1, \dots, m_s) \in \Z^{s+1} : 0 \leq m_i \leq \ell-1\}.\] 
	Let $k = \ell^{s+1} = |C_\ell|$, and let $Q_1, \dots, Q_k$ be the monomials represented by each point in $C_\ell$. Let $T$ be the set of monomials present in $Q_1P, \dots, Q_kP$. 
	Consider the Minkowski sum 
	\begin{align*}
		R + C_\ell = \{\bm{x} + \bm{y} : \bm{x}\in R, \,\bm{y} \in C_\ell\} \, .
	\end{align*}
	We have $T \subseteq R + C_\ell \subseteq C_{\ell + r}$ where $r$ is the largest component of any point of $R$. 
	
	We now view $Q_iP$, $1\leq i \leq k$, as LRS by assigning $x_0 = n+1$ as well as $x_j = \mu_j^n$ for all $1\leq j \leq s$. 
	They constitute $I_\ell$. 
	These~$k$ LRS are linearly independent. Indeed, their linear dependence would imply linear dependence of the LRS $u^{(i)}_n = Q_i(n+1,\mu_1^n, \dots, \mu_s^n)$, which is impossible as each $u^{(i)}_n = (n+1)^{m_0} \mu_1^{m_1 n} \dots \mu_s^{m_s n}$ for distinct tuples $(m_0, \dots, m_s)$. 
    The LRS corresponding to $Q_1P, \dots, Q_kP$ thus form a $\simskol(d,k)$ instance with
	\begin{align*}
		d \leq |T| \leq |R+C_\ell| \leq |C_{\ell + r}| = (\ell+r)^{s+1} \leq \ell^{s+1} + B \ell^s,
	\end{align*}
	where $B$ is a constant depending only on $r$ and $s$. 
	Note that the quotients of distinct characteristic roots of the LRS are not roots of unity since such quotients are always Laurent monomials of the $\mu_i$. 
	Indeed, let $\mu_1^{t_1}\dots\mu_s^{t_s}$ for $t_1, \dots, t_s \in \Z$ be a root of unity of order~$m > 0$. 
	Then $(mt_1, \dots, mt_s)$ needs to be a multiplicative relation of $\mu_1, \dots, \mu_s$ contradicting our assumption.
	
	
	Since $Q_i(n+1,\mu_1^n, \dots, \mu_s^n) \neq 0$ for all~$1 \leq i\leq k, \, n \in \N$, 
	the set of simultaneous zeros $\zero(I_\ell)$
	is exactly the set of~$n$ such that $P(n+1, \mu_1^n, \dots, \mu_s^n) = 0$.
	Therefore, it is exactly the set~$\zero(u)$.
\end{proof}

So, for $k$ arbitrarily large, the Skolem Problem on $u$ reduces to $\simskol(\tilde{d},k)$ with
\begin{align*}
	\tilde{d} \leq k + Bk^{\frac{s}{s+1}} \, ,
\end{align*}
for some constant $B$ depending on $u$, by solving an instance given by \cref{ssp-family} to decide if $\zero(u)\neq \varnothing$. Therefore, if $\simskol(d,k)$ is decidable for some $d \geq k + B k^{\frac{s}{s+1}} \geq \tilde{d}$, 
then one can determine whether $\zero(u)$ is non-empty.
This is employed in the proof of the hardness result next.

\begin{theorem}
	If there exists $C \in (0,1)$ such that $\orbprob(d,t)$ is decidable for all $(d,t)$ that satisfy~$t \leq Cd$, 
	then the Skolem Problem is decidable.
\end{theorem}
\begin{proof}
	Under the hypothesis, $\simskol(d,k)$ is decidable for all $d \geq \frac{1}{C}(d-k)$ by~\cref{prop:reduction}. 
	We now use this to decide for any non-degenerate LRS~$u$ whether $\zero(u)\neq\varnothing$. 
	Given~$u$ with a polynomial $P(x_0,\dots,x_s)$, 
	we shall find $d \geq k > 0$ such that 
	\begin{equation}\label{eq:hypo}
		\frac{1}{C}(d-k) \leq d
	\end{equation}
	while some instance of $\simskol(d,k)$ encodes the Skolem Problem for~$u$.
	As in~\cref{ssp-family}, consider $k = \ell^{s+1}$ as $\ell$ increases, and recall the constant~$B$ defined ibid.
	We set $d\coloneqq \ell^{s+1} + B \ell^s$.
	In order to satisfy~\eqref{eq:hypo}, we choose $\ell \geq \frac{1-C}{C}B$ 
	as we can take~$\ell$ arbitrary large. 
		
	
	$\simskol(d,k)$ is decidable by the hypothesis of the theorem.
	Let $\tilde{d}$ be the order of the LRS in the instance $I_\ell$ constructed in~\cref{ssp-family}.
	By construction, $d \geq \tilde{d}$ and so $\simskol(\tilde{d},k)$ is decidable, too.
	We can thus answer whether $\zero(u) \neq\varnothing$.
\end{proof}


\section{Complexity} \label{sec:complexity}
In this section, we will prove complexity upper bounds on the reduced Subspace Orbit Problem in the MSTV-reducible cases we have identified. 

We emphasise that unlike in previous sections, we consider $\orbprob_\KK$ over the field $\KK=\Q$.
Our input~$I$ will be a matrix $M \in \Q^{d\times d}$ such that no ratio of distinct eigenvalues of~$M$ is a root of unity, 
an initial vector $\bm x \in \Q^d$, 
and a linearly independent set of rational vectors $\{\bm v^{(1)}, \dots, \bm v^{(t)}\}$ spanning a target subspace $S \subseteq \Q^d$.
Let $\|I\|$ measure the bit-size of an input~$I$. 

Recall from~\cref{sec:intro} that an input to $\simskol_\Q(d,k)$ is $k$ linearly independent, non-degenerate LRS $\{u^{(1)}, \dots, u^{(k)}\}$ satisfying the same recurrence relation of order $d$; we assume the input is structured as
a vector of rational coefficients of the recurrence relation, 
along with~$k$ rational vectors of initial values.
We denote the sum of their bit-sizes by $\|(u^{(1)}, \dots, u^{(k)})\|$.

The notation $\poly(x)$ will denote any quantity bounded above by $x^{O(1)}$, and $2^{\poly(x)}$ any quantity bounded above by $2^{x^{O(1)}}$. 

To decide~$\orbprob_\Q(d,t)$ with $t \leq 2\log_3 d$, we proceed as follows:
\begin{steps}
	\item Show that $\orbprob_\Q$ on $(M,\bm x, S)$ reduces to $\simskol_\Q$ on $u^{(1)}, \dots, u^{(k)}$ such that $\|(u^{(1)}, \dots, u^{(k)})\| = \poly(\|(M,\bm{x}, S)\|)$. 
	\item If there is a linear combination of the LRS~$u^{(i)}$ that lies in the MSTV class, then such a linear combination exists, 
	whose coefficients (of the exponential polynomial form~\eqref{eq:exp-poly}) have bounded degree and height.
	\item Apply bounds on the zeros of non-degenerate LRS in the MSTV class to prove upper bounds on the simultaneous zeros of $u^{(1)},\dots, u^{(k)}$. 
\end{steps}



We note that Steps 1--3\ as listed justify the correctness of~\cref{cor:exp_zero_bound}, 
which provides an effective upper bound on the solutions~$n$ to $\orbprob_\Q$. 
The actual decision procedure amounts to the ``guess and check'' described in the proof of~\cref{thm:complex-end-result}. 
To this end, our complexity results require neither finding the sequences~$u^{(i)}$ (Step~1), nor explicitly computing an MSTV sequence (Step~2).

For Step~1, the reduction from~$\orbprob_\Q$ to $\simskol_\Q$, we need to prove that there is a basis of $S^\perp$ which is not too large. We use a version of Siegel's lemma, by Bombieri and Vaaler \cite[Theorem 2]{Bombieri1983}.
\begin{theorem} \label{thm:siegel}
Let $A \in \Z^{m \times n}$ be a matrix with $n>m$, whose rows are linearly independent. Suppose $|a_{ij}| < B$ for all $1 \leq i \leq m$ and $1 \leq j \leq n$. Then there are $n-m$ linearly independent integral solutions $\bm x^{(i)} = (x^{(i)}_1, \dots, x^{(i)}_n)^\top$ to
$A \bm x = 0$
such that 
\begin{align*}
    \prod_{i=1}^{n-m} \max_j |x^{(i)}_j| \leq G^{-1} \sqrt{|\det (A A^\top)|} \, ,
\end{align*}
where $G$ is the greatest common divisor of all the $m \times m$ minors of $A$.
\end{theorem}

\begin{lemma} \label{lem:basis_bound}
	Let $S$ be a subspace of~$\Q^d$ with $\dim S = d-k$. 
There is a basis $\bm w^{(1)}, \dots, \bm w^{(k)} \in \Q^d$ of $S^\perp$ such that $\|\bm w^{(i)}\| = \poly(\|S\|)$. 
\end{lemma}
\begin{proof}
$S$ is specified by a set of linearly independent rational vectors $\bm v^{(1)}, \dots, \bm v^{(t)} \in \Q^d$, where $t=d-k$. The lowest common multiple $L$ of the denominators of all the entries has modulus at most $2^{dt\|S\|}$. 
Let~$A$ be an integral $t \times d$ matrix
whose rows are $L \bm v^{(1)}, \dots, L \bm v^{(t)}$.
By \cref{thm:siegel}, there are $d-t$ linearly independent integral solutions $\bm w^{(1)}, \dots, \bm w^{(d-t)}$ to $A \bm x = 0$ that satisfy
\begin{align*}
    \prod_{i=1}^{d-t} \max_j | w^{(i)}_j| \leq \sqrt{|\det (A A^\top)|} \, .
\end{align*}
Then $\bm w^{(1)}, \dots, \bm w^{(d-t)}$ is exactly a basis of $S^\perp$. 
It remains to show \[\sqrt{|\det (A A^\top)|} \leq 2^{\poly(\|S\|)}\] to get the claimed result. By Hadamard's inequality,
\begin{align*}
    |\det (A A^\top)| &\leq (\max_{i,j} |(AA^\top)_{ij}|)^t t^{t/2} \leq (d \max_{i,j} |A_{ij}|^2)^t t^{t/2}  \\
    &\leq (dL^2)^t (\max_{i,j} |v^{(i)}_j|)^{2t} t^{t/2} = 2^{\poly(\|S\|)} \, .
\end{align*}
\end{proof}

\begin{lemma} \label{lem:char_poly_bound}
The characteristic polynomial $g_A$ of a matrix $A \in \Q^{d \times d}$ has coefficients of bit size $\poly(\|A\|)$. 
\end{lemma}
\begin{proof}
First, write $A = \frac{1}{a}B$, where $B$ has entries in $\Z$, noting we can choose $a \neq 0$ such that $\|a\|,\|B\| = \poly(\|A\|)$. 
Let $g_B(x) = x^d + c_1 x^{d-1} + \dots + c_d$. If $J \subseteq \{1, \dots , d\}$ then let $B_J$ denote the principal submatrix formed of rows and columns with indices in~$J$. Then the coefficients $c_r$ may be expressed as a sum of principal minors of $B$ (see e.g. \cite[p.294]{meyer_matrix_2023}) as
\begin{align*}
    c_r = \sum_{|J| = r} (-1)^r \det(B_J) \, .
\end{align*}
Write $B = (b_{ij})_{1 \leq i,j \leq d}$. There are $\binom{d}{r}$ many $r \times r$ principal minors of $A$ so by Hadamard's inequality applied to $B_J$ we get
\begin{align*}
    |c_r| \leq \binom{d}{r} \left(\max_{1 \leq i , j \leq d} |b_{ij}| \right)^r r^{r/2} \leq d! \|B\|^d d^{d/2}
\end{align*}
and therefore $\|c_r\| = \poly(\|A\|)$. Now note that $g_B(a\cdot A) = 0$ and after dividing by $a^d$, it is easy to see that
$g_A(x)$ is
\begin{align*}
    g_A(x) =x^d + \frac{c_1}{a} x^{d-1} + \dots + \frac{c_d}{a^d} \, ,
\end{align*}
so indeed every coefficient $c_r/a^{r}$, $1 \leq r \leq d$, has bit-size $\poly(\|A\|)$, since $\|a\|,\|c_r\| = \poly(\|A\|)$, for all $1 \leq r \leq d$.
\end{proof}
The proposition below refines a reduction from~\cref{prop:reduction}.
\begin{proposition} \label{prop:Orbit_to_SSP}
$\orbprob_\Q$ on $(M,\bm x, S)$ reduces to $\simskol_\Q$ on $u^{(1)}, \dots, u^{(k)}$, for $\|(u^{(1)}, \dots, u^{(k)})\| = \poly(\|(M,\bm x, S)\|)$. 
\end{proposition}
\begin{proof}
Let $\bm w^{(1)}, \dots, \bm w^{(k)}$ be a basis for $S^\perp$. Then $\orbprob_\Q$ on $(M, \bm x , S)$ reduces to $\simskol_\Q$ on the LRS $u^{(i)}_n = (\bm{w}^{(i)})^\top M^n \bm x$ for $1 \leq i \leq k$ (see~\cref{prop:reduction}). 
By~\cref{lem:basis_bound}, we have a basis of~$S^\perp$ such that $\|\bm w^{(i)}\| = \poly(\|S\|)$ for all $i$, so we may bound the initial values $\|u^{(i)}_n\| = \poly(\|(M,\bm x , S)\|)$ for each $0 \leq n \leq d-1$ and $1 \leq i \leq k$. 
Also, the characteristic polynomial of $M$ is exactly the characteristic polynomial of a recurrence relation~$(c_1, \dots, c_d) \in \Q^d$ satisfied by each $u^{(i)}$, so by \cref{lem:char_poly_bound}, we have $\|u^{(i)}\| = \poly(\|(M,\bm x , S)\|)$. 
\end{proof}

We have completed Step~1. Next, we discuss the exponential polynomial form~\eqref{eq:exp-poly} of the LRS, and so we need to bound the heights and degrees of the coefficients and characteristic roots.
\begin{proposition} \label{prop:LRS_coeff_bounds}
Given an order-$d$ $\Q$-LRS $u = \sum_{i,j} c_{i,j} n^j \lambda_i^n$, we have $h(c_{i,j}), h(\lambda_i) = \poly(\|u\|)$, and $\deg(\lambda_i) \leq d$ and $c_{i,j} \in \Q(\lambda_i)$.
\end{proposition}
\begin{proof}
Since $\lambda_i$ is a root of the characteristic polynomial
\begin{align*}
    g(x) = a_dx^d + a_{d-1} x^{d-1} + \dots + a_0
\end{align*}
of $u$, where $a_i \in \Z$, we have $\deg \lambda_i \leq d$. Furthermore, let $f(x) = b_{r} x^{r} + \dots + b_0$ be the minimal polynomial of $\lambda$, and let $H(\lambda)$ denote the ``naïve height'' of $\lambda$, defined by $H(\lambda) = \max\{|b_0|,\dots , |b_r|\}$. Then, since $f$ divides $g$, by \cite[Lemma 3.11 and Remark 2, p. 81]{Waldschmidt_book} we have
\begin{align*}
    h(\lambda_i) &\leq \frac{1}{r} \log H(\lambda_i) + \frac{1}{2r} \log(r+1) \\
    &\leq \frac{1}{r} \log \left(2^d\sqrt{d+1} H(g) \right) + \frac{1}{2r} \log(r+1)\\
    &= \poly(\|u\|) \, .
\end{align*}
Now, by \cite[Lemma 6]{akshay_2020}, we have $c_{i,j} = q_0 + q_1 \lambda_i + \dots + q_{d-1} \lambda_i^{d-1}$ for rationals $|q_i| \leq 2^{\poly(\|u\|)}$. Therefore, using~\cref{prop:h_props}, 
\begin{align*}
h(c_{i,j}) \leq \log d + \sum_{s=0}^{d-1} h(q_s) + \sum_{s=0}^{d-1} sh(\lambda_i) = \poly(\|u\|) \,
\end{align*} holds as well as $\deg( c_{i,j}) \leq d$.
\end{proof}

We proceed with Step~2 of our analysis: 
if there is a linear combination of $u^{(1)}$, \dots, $u^{(k)}$ lying in the MSTV class, then such a linear combination exists of small height and degree. 
First, we have a lemma that allows us to restrict ourselves to kernels of dimension~1, which simplifies the proof. 
Essentially, the lemma says that if one has a linear combination in the MSTV class with minimal support amongst linear combinations in the MSTV class, then in fact the support is minimal amongst all linear combinations in general. 
\begin{lemma} \label{lem:min_supp}
Let $C$ be the span over $\Alg$ of $\bm u^{(1)}, \dots, \bm u^{(k)} \in \Alg^d$, the vectors representing the LRS $u^{(1)}, \dots, u^{(k)}$. 
Suppose that $\bm u \in C$ represents an LRS in the MSTV class 
and has a minimal support among such vectors in~$C$;
that is, if $\bm u' \in C $ represents an LRS in the MSTV class then 
$\supp(\bm u') \subsetneq \supp(\bm u)$ implies $\bm u' = \bm 0$. 

	Then we have $\dim (\ker \pi_{\neg V} |_C) = 1$ where $V = \supp(\bm u)$
	and $\pi_{\neg V}$ is the projection to the components $\{1,\dots, d\} \setminus V$. 
	That is, any $\bm w \in C$ with $\supp(\bm w) \subseteq \supp(\bm u)$ is of the form $\bm w = \eta \bm u$ for some $\eta \in \Alg$.
\end{lemma}
\begin{proof}
Since $\bm u$ corresponds to an LRS in the MSTV class, it corresponds to an LRS with at most $r$ dominant roots with respect to some absolute value $|\cdot|_v$, where $r \leq 2$ if $|\cdot|_v$ is non-Archimedean and $r \leq 3$ if $|\cdot|_v$ is Archimedean. 
Suppose $\supp(\bm w) \subseteq \supp(\bm u)$. If $\bm w \neq \eta \bm u$ for any $\eta \in \Alg$, then 
we can choose~$i$ such that $i \in \supp(\bm u)$ corresponds to a $|\cdot|_v$-dominant term, as well as choose~$j \neq i$,
so that $\frac{w_i}{u_i} \neq \frac{w_j}{u_j}$. 
Then $\bm w' = \bm w- \frac{w_j}{u_j} \bm u$ has $w'_i \neq 0$ and thus the set of $|\cdot|_v$-dominant terms for~$\bm w'$ is a subset of $|\cdot|_v$-dominant terms for~$\bm w$. 
Therefore, the LRS of $\bm w'$ lies in the MSTV class.
Moreover, $w_j' = 0$ and so $\supp(\bm w') \subsetneq \supp(\bm u)$,
which contradicts the fact that $\supp(\bm u)$ is minimal and concludes the proof.
\end{proof}
This lemma allows us to construct a linear combination in the MSTV class of small height. First, we need an elementary height bound on the determinant of a matrix.
\begin{restatable}{lemma}{hdetbound} \label{lem:h_det_bound}
For a matrix $A = (a_{ij})_{1 \leq i , j \leq m} \in \KK^{m \times m}$ over a number field~$\KK$, we have
\begin{align*}
    h(\det A) \leq \frac{m}{2}\log m + m \sum_{1 \leq i , j \leq m} h(a_{ij}) \, .
\end{align*}
\end{restatable}
The proof is elementary, and follows from Hadamard's inequality and the definition of the height; it may be found in Appendix \ref{sec:appendix}.

\begin{lemma} \label{lem:small_MSTV_height}
Suppose there exists $\bm \beta = (\beta_1, \dots, \beta_k)\in \Alg^k$ such that $\sum_{i=1}^k \beta_i u^{(i)}$ lies in the MSTV class. Then there exists such a $\bm \beta$ with $h(\beta_i) = \poly(\|u\|)$. 
\end{lemma}
\begin{proof}
If there exists a linear combination in the MSTV class, then certainly there exists 
a linear combination~$\bm v$
with minimal support among linear combinations in the MSTV class. 
By \cref{lem:min_supp}, if $\supp(\bm v) = V$, then $\dim(\ker \pi_{\neg V}|_C) = 1$: 
any vector in $\ker \pi_{\neg V}|_C$ is a scalar multiple of $\bm v$ and lies in the MSTV class. Let
\begin{align*}
    A = \begin{pmatrix}
        \pi_{\neg V}(\bm u^{(1)}) \mid \dots \mid \pi_{\neg V}(\bm u^{(k)})
    \end{pmatrix} \, ,
\end{align*}
then we just need to find $\bm \beta$ such that $A \bm \beta = 0$. 

Since $\dim \ker A = 1$, $\mathrm{rank} \ A = k-1$, so we can pick $k-1$ linearly independent rows of $A$ to form a $(k-1) \times k$ matrix $B$ with the same kernel. Let $\beta_i = (-1)^{i+1} \det B^{(i)}$, where $B^{(i)}$ is $B$ with the $i$th column removed. This is an element of the kernel; indeed, we have by Laplace expansion that for any vector $\bm y = (y_1, \dots, y_k)^\top \in \Alg^k$,
\begin{align*}
    \det \begin{pmatrix}
        \bm y^\top \\ B 
    \end{pmatrix}
    = \sum_{i=1}^k (-1)^{i+1} y_i \det B^{(i)} = \bm y^\top \bm \beta
\end{align*}
and so $\bm y^\top \bm \beta = 0$ if $\bm y^\top$ is any row of $B$, so $B \bm \beta = 0$ and hence $A \bm \beta = 0$. Now by \cref{lem:h_det_bound}, we have $h(\beta_i) = \poly(\|u\|)$, as required.
\end{proof}
\begin{restatable}{corollary}{smallmstv}\label{cor:small_MSTV_LRS}
Suppose $\bm \beta = (\beta_1, \dots, \beta_k)\in \Alg^k$ is such that $\sum_{i=1}^k \beta_i u^{(i)} = \sum_{i=1}^R P_i(n) \lambda_i^n$ lies in the MSTV class and is of minimal support. Let $\KK = \Q(\lambda_1, \dots, \lambda_R)$. Then there exists $\bm \gamma \in \KK^k$ such that $\bm \gamma = \eta \bm \beta$ for some $\eta \in \Alg$ with $h(\bm \gamma) = \poly(\|\bm \beta\|)$.
\end{restatable}

The proof of this corollary follows from using the trace to construct $\bm \gamma$; it is given in Appendix~\ref{sec:appendix}. There, we also prove
the following bound on zeros of LRS.
\begin{restatable}{theorem}{MSTVbound} \label{thm:MSTV_zero_bound}
Suppose that $u_n = \sum_{i=1}^s P_i(n) \lambda_i^n$, for polynomials $P_i$ of degree at most $\delta$. 
Let $h_P$ be the largest height of any coefficient of any $P_i$ and let $h_\lambda \coloneq \max_{i=1}^s h(\lambda_i)$.
Suppose that $\LRS{u}$ is non-degenerate and lies in the MSTV class, that is, there is absolute value $|\cdot|_v$ with
\begin{align*}
    |\lambda_1|_v = \dots = |\lambda_r|_v > |\lambda_{r+1}|_v \geq \dots \geq |\lambda_s|_v
\end{align*}
and $r \leq 3$ if $|\cdot|_v$ is Archimedean, and $r \leq 2$ if $|\cdot|_v$ is non-Archimedean. Let $\KK$ be the number field generated by $\lambda_1, \dots, \lambda_s$ and the coefficients of $P_1, \dots, P_s$, and let $D = [\KK:\Q]$. Then $u_n = 0$ implies that
\begin{align*}
    n < \poly(s) \cdot 2^{\poly(D,h_\lambda,h_P,\delta)} \, .
\end{align*}
\end{restatable}
The proof is essentially the same as the decidability proof for LRS in the MSTV class using Baker's theorem \cite{mignotte_logarithm_1994,vereshchagin_1985,bilu_skolem_2025}, but we carry through the constants more precisely to achieve the desired bound (\cite[Sections C-F]{chonev_complex2016} proves a special case for order $\leq 4$ LRS). 

As Step~3, we apply this to derive upper bounds on the simultaneous zeros of LRS; 
and hence on the solutions to~$\orbprob_\Q(d,t)$.
\begin{corollary} \label{cor:exp_zero_bound}
If $M \in \Q^{d \times d}$, $\bm x \in \Q^d$ and $S \subseteq \Q^d$ of dimension~$t$ define an MSTV-reducible instance of $\orbprob_\Q(d,t)$,
then $M^n \bm x \in S$ implies that $n < 2^{\poly(d^t,\|(M,\bm x , S)\|)}$.
\end{corollary}
\begin{proof}
By combining \cref{prop:Orbit_to_SSP,prop:LRS_coeff_bounds}, \cref{lem:small_MSTV_height,cor:small_MSTV_LRS}, 
there is a non-degenerate LRS lying in the MSTV class defined by
\begin{align*}
    u_n = \sum_{i=1}^R P_i(n) \lambda_i^n = \sum_{i=1}^R \sum_{j=0}^{\deg P_i} c_{i,j} n^j \lambda_i^n
\end{align*}
such that if $M^n \bm x \in S$ then $u_n=0$, and furthermore $h(\lambda_i),h(c_{i,j}) = \poly(\|(M,\bm x , S)\|)$, $\deg(\lambda_i) \leq d$ and $c_{i,j} \in \Q(\lambda_1, \dots, \lambda_R)$ for all $i,j$. 
Moreover, $u$ arises as a linear combination of $d-t$ linearly independent LRS such that $u$ has minimal support, meaning that $R \leq t+1$. 
Therefore, letting $\KK = \Q(\lambda_1, \dots, \lambda_R)$, we have $D = [\KK : \Q] \leq d^{t+1}$. Finally, \cref{thm:MSTV_zero_bound} implies that if $u_n = 0$, then $n < 2^{\poly(d^t,\|(M,\bm x , S)\|)}$.
\end{proof}
Now we may prove our main complexity result. Before stating it, recall that $\EqSLP$ is the complete class for the problem of checking whether a polynomial-size circuit evaluates to zero. It is known that $\EqSLP \subseteq \coRP$ \cite{Schonhage_1979,allender_complexity_2009}.
\begin{theorem} \label{thm:complex-end-result}
For every fixed constant $T \in \N$, when restricting to $t \leq T$, we have that $\orbprob_\Q(d,t)$ for all $(d,t)$ such that $t \leq 2\log_3 d$ is in $\NP^{\EqSLP} \subseteq \NP^\RP$. 
Furthermore, for every fixed constant $D \in \N$, when restricting to $d \leq D$, the problem is in $\coRP$.
\end{theorem}
\begin{proof}
	Let $M \in \Q^{d \times d}$, $\bm x\in \Q^d$, and target subspace $S \subseteq \Q^d$ with $\dim S = t$ be given. Since we assume $t \leq T$ for fixed $T \in \N$, we have $d^t = \poly(d) = \poly(\|M\|)$. 
By \cref{cor:exp_zero_bound}, there is an upper bound
\begin{align*}
    N = 2^{\poly(d^t,\|(M,\bm x , S)\|)} = 2^{\poly(\|(M,\bm x, S)\|)}
\end{align*}
such that if $M^n \bm x \in S$, then $n \leq N$. Also, one may compute a basis $\bm w^{(1)}, \dots, \bm w^{(d-t)}$ of $S^\perp$ in polynomial time. Let $u^{(i)} = (\bm{w}^{(i)})^\top M^n \bm x$. We present an algorithm which then puts $\orbprob_\Q(d,t)$ in~$\NP^\RP$:
\begin{enumerate}
    \item Guess an index $0 \leq n \leq N$;
    \item Check whether $u^{(i)}_n = 0$ for all $1 \leq i \leq d-t$ as follows: 
    build a polynomial-size circuit to compute $(\bm{w}^{(i)})^\top M^n \bm x$ for each $1 \leq i \leq d-t$, by using repeated squaring on $M$. 
\end{enumerate} 
Now suppose that $d$ is bounded. 
Since $N = 2^{\poly(\|(M,\bm x, S)\|)}$, we have $\|N\| = \poly(\|(M,\bm x , S)\|)$. Therefore, $\orbprob_\Q(d,t)$ reduces to the Bounded Skolem Problem (as defined in \cite{bacik_complexity_2026}) with input LRS $u_n = (u^{(1)}_n)^2 + \dots + (u_n^{(d-t)})^2$ and integer $N$. It is shown in \cite{bacik_complexity_2026} that when the order of the input LRS is bounded, the Bounded Skolem Problem lies in $\coRP$. 
Since the order of $u$ is bounded by a function in $d$, we thus have that the $\orbprob_\Q(d,t)$ with bounded ambient dimension $d$ is in $\coRP$.
\end{proof}

\bibliography{main.bib}

\clearpage

\appendix
\section{Missing Proofs} \label{sec:appendix}
\subsection{Bound on height of determinant}
\hdetbound*
\begin{proof}
\begin{align*}
h(\det A) &= \frac{1}{[\KK : \Q]} \sum_{\substack{v \in M_\KK \\ v \mid \infty }} D_v \log^+ |\det A|_v  \\ 
&+ \frac{1}{[\KK : \Q]} \sum_{\substack{v \in M_\KK \\ v \nmid \infty }} D_v \log^+ |\det A|_v \, .
\end{align*}
First we bound the sum over Archimedean absolute values:
\begin{align}
    &\frac{1}{[\KK : \Q]} \sum_{\substack{v \in M_\KK \\ v \mid \infty }} D_v \log^+ |\det A|_v \notag \\
    &\leq \frac{1}{[\KK : \Q]} \sum_{\substack{v \in M_\KK \\ v \mid \infty }} D_v \log^+\left( \left( \max_{i,j} |a_{ij}|_v \right)^{m} m^{\frac{m}{2}}\right) \label{eqn:h_det_bound1} \\
    &\leq \frac{1}{[\KK : \Q]} \sum_{\substack{v \in M_\KK \\ v \mid \infty }} D_v \frac{m}{2}  \log m + \frac{1}{[\KK : \Q]} \sum_{\substack{v \in M_\KK \\ v \mid \infty }} D_v \sum_{1 \leq i,j \leq m} m\log^+ |a_{ij}|_v  \notag \\
    &\leq \frac{m}{2} \log m + \frac{m}{[\KK : \Q]} \sum_{\substack{v \in M_\KK \\ v \mid \infty }} D_v \sum_{1 \leq i,j \leq m} \log^+ |a_{ij}|_v \label{eqn:h_det_bound12} 
\end{align}
where \eqref{eqn:h_det_bound1} arises from Hadamard's inequality applied to each embedding of $A$ into $\CC$. Now, we bound the sum over non-Archimedean~$v$:
\begin{align}
&\frac{1}{[\KK : \Q]} \sum_{\substack{v \in M_\KK \\ v \nmid \infty }} D_v \log^+ |\det A|_v \notag \\ &
\leq \frac{1}{[\KK : \Q]} \sum_{\substack{v \in M_\KK \\ v \nmid \infty }} D_v \log^+ \left( \max_{\sigma \in S_m} \prod_{i=1}^m |a_{i,\sigma(i)}|_v \right) \label{eqn:h_det_bound21}\\
&\leq \frac{1}{[\KK : \Q]} \sum_{\substack{v \in M_\KK \\ v \nmid \infty }} D_v \sum_{1 \leq i,j \leq m} \log^+ |a_{ij}|_v \label{eqn:h_det_bound2}
\end{align}
where \eqref{eqn:h_det_bound21} is due to the Leibniz formula for determinants and the ultrametric inequality. 
With~\eqref{eqn:h_det_bound12} and~\eqref{eqn:h_det_bound2} we get
\begin{align*}
    h(\det A) &\leq \frac{m}{2}\log m + \frac{m}{[\KK : \Q]} \sum_{v \in M_\KK} D_v \sum_{1 \leq i,j \leq m} \log^+ |a_{ij}|_v \\
    &\leq \frac{m}{2}\log m + m \sum_{1 \leq i , j \leq m} h(a_{ij}) \, ,
\end{align*}
as required.
\end{proof}
\subsection{Height bound for a linear combination}
Given a finite field extension $\LL/\KK$ of number fields, we may define a linear map $U_x : \LL \to \LL$ by $U_x(y) = xy$. This is a $\KK$-linear map, so it has a trace $\mathrm{Tr}(U_x) \in \KK$. We thus define the \emph{trace of $x$ relative to the extension $\LL/\KK$} as $\mathrm{Tr}_{\LL/\KK}(x) = \mathrm{Tr}(U_x)$. If $\sigma_1, \dots, \sigma_n$ are all the distinct $\KK$-embeddings of $\LL$ into the algebraic closure of $\KK$, we have 
\begin{align*}
    \mathrm{Tr}_{\LL / \KK}(x) = \sum_i \sigma_i(x) \, .
\end{align*}
The map $(x,y) \mapsto \mathrm{Tr}_{\LL/\KK}(xy)$ defines a non-degenerate bilinear form on the $\KK$-vector space $\LL$ \cite[Prop. 2.8]{neukirch_algebraic_1999}, meaning that if $x \in \LL$ has the property that $Tr_{\LL/\KK}(xy) = 0$ for all $y \in \LL$ then $x = 0$.
\smallmstv*
\begin{proof}
Let $\LL = \Q(\lambda_1, \dots, \lambda_s)$, and assume without loss of generality that $\beta_1 \neq 0$. Then for any $\mu \in \LL$, consider
\begin{align*}
    \mathrm{Tr}_{\LL / \KK} \left( \mu \sum_{i=1}^k \beta_i u^{(i)}_n \right) = \sum_{i=1}^R \mathrm{Tr}_{\LL / \KK} ( \mu P_i(n)) \lambda_i^n \, .
\end{align*}
Since the trace is a non-degenerate bilinear form, there exists $\mu \in \LL$ such that $\mathrm{Tr}_{\LL / \KK} (\mu P_i(n))$ are not all identically zero. Thus, 
\begin{align*}
    \mathrm{Tr}_{\LL / \KK} \left( \mu \sum_{i=1}^k \beta_i u^{(i)}_n \right) = \sum_{i=1}^k \mathrm{Tr}_{\LL / \KK}(\mu \beta_i) u^{(i)}_n
\end{align*}
is another non-zero linear combination of the $u^{(i)}$ whose support is a subset of the support of $\sum_{i=1}^k \beta_i u^{(i)}$. 
By minimality and \cref{lem:min_supp}, there exists $\eta \in \LL$ such that $\mathrm{Tr}_{\LL / \KK}( \mu  \beta_i) = \eta \beta_i$ for all $1 \leq i \leq k$. 
We have $\eta \beta_i \in \KK$ for all $1 \leq i \leq k$, and hence 
$\frac{\beta_i}{\beta_1} \in \KK$ for all $1 \leq i \leq k$. 
Take $\bm \gamma \coloneq \frac{1}{\beta_1} \bm \beta \in \KK^k$, so $h(\gamma_i) = \poly(\|\beta_i\|)$ for all $1 \leq i \leq k$.
\end{proof}
\subsection{Bounds on zeros of LRS}
In this section, we prove upper bounds on the zeros of non-degenerate LRS in the MSTV class. We essentially reprove the results of \cite{tijdeman_mst1984,vereshchagin_1985}, but carry the constants through in more detail so that the dependence on each quantity is clear. The structure of our proof more or less follows the exposition given by \cite{bilu_skolem_2025}, though we provide a little more detail when necessary. 

The proofs will rely on Baker's theorem on linear forms in logarithms, and subsequent improvements and extensions to the $p$-adic case. For the case of the modulus, we have the following formulation of Matveev \cite[Corollary 2.3]{matveev_explicit_2000}. In what follows we shall always take the principal branch of $\log$ on $\CC \setminus \{0\}$, that is we have 
\begin{align*}
- \pi < \arg z = \mathrm{Im} \log z \leq \pi \, .
\end{align*}

\begin{theorem}[Matveev]\label{thm:Matveev} Let $\alpha_1, \dots , \alpha_s$ be non-zero complex algebraic numbers contained in a number field of degree $D$. Let $b_1, \dots , b_s \in \Z$ be such that 
\begin{align*}
    \Lambda = b_1 \log \alpha_1 + \dots + b_s \log \alpha_s \neq 0 \, .
\end{align*}
Further, let $A_1, \dots , A_s, B$ be real numbers such that 
\begin{align*}
    &A_j \geq \max\{Dh(\alpha_j), |\log \alpha_j|, 0.16\} \, \, \, (j = 1, \dots , s) \\
    &B = \max\{|b_1|, \dots , |b_s|\} \, .
\end{align*}
Then 
\begin{align*}
    |\Lambda| \geq \exp\left(-2^{6s+20}D^2 A_1 \dots A_s (1+ \log D)(1+ \log B)\right) \, .
\end{align*}
\end{theorem}

In the $p$-adic case, we have the following theorem of Yu \cite[Theorem 1]{yu_p-adic_1999}.

\begin{theorem}[Yu]\label{thm:Yu}
Let $\alpha_1, \dots , \alpha_s$ be non-zero algebraic numbers contained in a number field $K$ of degree $D$. Let $\pp \subseteq \O_K$ be a prime ideal lying above integer prime $p$, such that $v_\pp(\alpha_i) = 0$ for $i = 1, \dots, s$. Let $b_1,\dots,b_s \in \Z$ be such that
\begin{align*}
    \Xi = \alpha_1^{b_1} \cdots \alpha_s^{b_s}-1 \neq 0 \, .
\end{align*}
Further, let $A_1, \dots , A_s, B$ be real numbers such that
\begin{align*}
    &A_j \geq \max\{h(\alpha_j), \log p\} \, \, \, (j=1 , \dots , s) \\
    &B \geq \max \{|b_1| , \dots , |b_s|, 3\}.
\end{align*}
Then
\begin{align*}
    v_\pp(\Xi) \leq C(s,D,p)A_1 \cdots A_s \log B
\end{align*}
where we may take
\begin{align*}
C(s,D,p) \leq 4ca^se^s \left( \frac{p}{p-1} \right)^s (s+1)^{s+3}(p^D-1)D^{2s+2}  \max\{D\log p, \log(e^4(s+1)D)\} 
\end{align*}
for $a= 32, c = 712$.
\end{theorem}
Note that the form of $C(s,D,p)$ we give is simplified from the original statement --- the important thing for our result is that $C(2,D,p) = \poly(p^D)$.
Our main goals are to prove the following results, from which we shall deduce the required upper bounds on the zeros of non-degenerate LRS in the MSTV class. 

\begin{proposition} \label{prop:2_sum}
Let $\alpha_1,\alpha_2,\lambda_1,\lambda_2 \in \KK$, a number field of degree $D$, such that $\lambda_1/\lambda_2$ is not a root of unity. Let $h_\alpha = \max\{h(\alpha_1),h(\alpha_2)\}, h_\lambda = \max\{h(\lambda_1), h(\lambda_2)\}$. For $n > \poly(D,h_\alpha)$, when $|\cdot|_v$ is Archimedean, we have
\begin{align*}
    |\alpha_1 \lambda_1^n + \alpha_2 \lambda_2^n|_v \geq |\lambda_1|_v^n e^{-C_1 (1+ \log n)}
\end{align*}
and when $|\cdot|_v = |\cdot|_\pp$ is non-Archimedean, we have
\begin{align*}
    |\alpha_1 \lambda_1^n + \alpha_2 \lambda_2^n|_\pp \geq |\lambda_1|_\pp^n e^{-C_2(1+ \log n)}
\end{align*}
for prime ideal $\pp \subseteq \O_\KK$ lying above prime $p$ and some constants $C_1,C_2 > 0$, which we may take to be
\begin{align*}
    C_1 &= \poly(D,h_\alpha,h_\lambda) \\
    C_2 &= \poly(p^D,D,h_\alpha,h_\lambda) \, .
\end{align*}
\end{proposition}

\begin{proposition} \label{prop:3_sum_nonzero}
Let $\alpha_1,\alpha_2,\alpha_3,\lambda_1,\lambda_2,\lambda_3$ be non-zero algebraic numbers in a number field $\KK$ of degree $D$, such that no quotient $\lambda_i/\lambda_j$ is a root of unity for $i \neq j$. Let $h_\alpha = \max\{h(\alpha_1),h(\alpha_2),h(\alpha_3)\}$ and $h_\lambda = \max\{h(\lambda_1),h(\lambda_2),h(\lambda_3)\}$. Suppose there is an Archimedean absolute value $|\cdot|_v$ such that
\begin{align*}
    |\lambda_1|_v = |\lambda_2|_v = |\lambda_3|_v \, .
\end{align*}
Suppose also that
\begin{align*}
    \alpha_1 \lambda_1^n + \alpha_2 \lambda_2^n + \alpha_3 \lambda_3^n = 0 \, .
\end{align*}
Then $n \leq \poly(h_\alpha)2^{\poly(D,h_\lambda)}$.
\end{proposition}

\begin{proposition} \label{prop:3_sum}
Let $\alpha_1,\alpha_2,\alpha_3,\lambda_1,\lambda_2,\lambda_3$ be non-zero algebraic numbers in a number field $\KK$ of degree $D$ such that no quotient $\lambda_i/\lambda_j$ is a root of unity for $i \neq j$. Let $h_\alpha = \max\{h(\alpha_1),h(\alpha_2),h(\alpha_3)\}$ and $h_\lambda = \max\{h(\lambda_1),h(\lambda_2),h(\lambda_3)\}$. Let $|\cdot|_v$ be an Archimedean absolute value such that
\begin{align*}
    |\lambda_1|_v= |\lambda_2|_v = |\lambda_3|_v \, .
\end{align*}
Let $n \in \N$ with $n \geq 3$, such that 
\begin{align} \label{eqn:3_term_nonzero}
    \alpha_1 \lambda_1^n + \alpha_2 \lambda_2^n + \alpha_3 \lambda_3^n \neq 0 \, .
\end{align}
Then
\begin{align*}
    |\alpha_1 \lambda_1^n + \alpha_2 \lambda_2^n + \alpha_3 \lambda_3^n|_v \geq |\lambda_1|_v^n e^{-C(1+ \log n) }
\end{align*}
where $C = \poly(D,h_\alpha,h_\lambda)$.
\end{proposition}

First, we prove \cref{prop:2_sum}, but for this, we need a lemma.
\begin{lemma} \label{lem:Voutier_cor}
Suppose $x,y \in \KK$ such that neither $x$ nor $y$ is a root of unity and there is $n \in \N$ with $y = x^n$. Then $n \leq \poly(D) h(y)$.
\end{lemma}
\begin{proof}
Follows from the fact that $y = x^n$ implies $n = \frac{h(y)}{h(x)}$, and a result of Voutier \cite[Corollary 2]{Voutier_1996} that shows $h(x) \geq \frac{2}{D(\log (3D))^3}$ when $D \geq 2$, and that $h(x) \geq \log 2$ for any $x \in \Q \setminus\{1,-1\}$, that is, when $D = 1$.
\end{proof}

\begin{proof}[Proof of Proposition \ref{prop:2_sum}]
Let $\Xi = \frac{-\alpha_2}{\alpha_1} \left( \frac{\lambda_2}{\lambda_1} \right)^n - 1$. By Lemma \ref{lem:Voutier_cor}, $\Xi \neq 0$ for $n \geq \poly(D) h(\frac{\alpha_2}{\alpha_1}) = \poly(D,h_\alpha)$. 

If $|\cdot|_v = |\cdot|_\sigma$ for some embedding $\sigma : \KK \hookrightarrow \CC$, we may replace $x$ by $\sigma(x)$ for $x = \alpha_1,\alpha_2,\lambda_1, \lambda_2$ and use $h(\sigma(x)) = h(x)$, so without loss of generality we may assume $|\cdot|_v = |\cdot|$. Note now that for any $|z-1| \leq \frac{1}{2}$, we have $|\log z|  \leq (2 \log 2)|z-1|$. Indeed,
\begin{align*}
    \left| \frac{\log z}{z-1} \right| = \left| \sum_{j=0}^\infty (-1)^j \frac{(z-1)^j}{j+1} \right| \leq \sum_{j=0}^\infty \frac{2^{-j}}{j+1} \leq 2 \log 2 \, 
\end{align*}
which gives the claimed inequality. Therefore, we get
\begin{align*}
|\Xi| &\geq \frac{1}{2 \log 2} \left|\log \left(\frac{-\alpha_2}{\alpha_1} \left( \frac{\lambda_2}{\lambda_1} \right)^n \right) \right| \\ &= \frac{1}{2 \log 2} \left|\log \left(\frac{-\alpha_2}{\alpha_1} \right) + n\log \left( \frac{\lambda_2}{\lambda_1} \right) - 2m\pi i\right|
\end{align*}
for some integer $m \in \Z$ with $|2m| \leq n+2$. The right-hand side may now be bounded from below by Theorem \ref{thm:Matveev} using $\pi i = \log(-1)$. This gives 
\begin{align*}
|\Xi| \geq \exp \left(-CD^2 h\left( \frac{\alpha_2}{\alpha_1} \right) h\left( \frac{\lambda_2}{\lambda_1} \right)(1+ \log D)(1+ \log n) \right)
\end{align*}
for some absolute constant $C$. Multiply through by $|\alpha_1 \lambda_1^n|$ to get
\begin{align*}
&|\alpha_1 \lambda_1^n + \alpha_2 \lambda_2^n | \geq |\lambda_1|^n |\alpha_1| |\Xi| \\
&\geq |\lambda_1|^n e^{-Dh(\alpha_1)} \exp \left(-CD^2 h\left( \frac{\alpha_2}{\alpha_1} \right) h\left( \frac{\lambda_2}{\lambda_1} \right)(1+ \log D)(1+ \log n) \right) \, .
\end{align*}

In the non-Archimedean case, note first that if 
\begin{align*}
    \left| \frac{-\alpha_2}{\alpha_1} \left( \frac{\lambda_2}{\lambda_1} \right)^n \right|_\pp \neq 1 
\end{align*}
then $|\Xi|_\pp \geq 1$, so the claimed inequality follows trivially. Therefore, we may assume for the rest of this proof that
\begin{align} \label{eqn:absval_1}
    \left| \frac{-\alpha_2}{\alpha_1} \left( \frac{\lambda_2}{\lambda_1} \right)^n \right|_\pp = 1 \, . 
\end{align}
Suppose that $\left| \frac{\lambda_2}{\lambda_1}\right|_\pp \neq 1$, then by \eqref{eqn:absval_1} we have
\begin{align*}
    n = \frac{\log \left|\alpha_1 / \alpha_2 \right|_\pp}{\log \left| \lambda_2/\lambda_1 \right|_\pp } \, .
\end{align*}
Note that $\left| \frac{\lambda_2}{\lambda_1}\right|_\pp \neq 1$ implies that either $\left| \frac{\lambda_2}{\lambda_1}\right|_\pp \geq p^{1/e_\pp}$ or $\left| \frac{\lambda_2}{\lambda_1}\right|_\pp \leq p^{- 1/e_\pp}$. From this and Proposition \ref{prop:h_props} item \ref{enum:h_props_log} it follows that
\begin{align*}
    n \leq \frac{e_\pp}{\log p}Dh(\alpha_2/\alpha_1) \leq \poly(D,h_\alpha) \, .
\end{align*}
So certainly for $n > \poly(D,h_\alpha)$ we have $|\Xi|_\pp = 1$.

Otherwise, we have $\left| \frac{\lambda_2}{\lambda_1}\right|_\pp = 1$ and so $\left| \frac{\alpha_2}{\alpha_1}\right|_\pp = 1$. Thus, whenever $|\Xi|_\pp \neq 0$, we may apply Theorem \ref{thm:Yu} to get
\begin{align*}
&|\alpha_1 \lambda_1^n + \alpha_2 \lambda_2^n|_\pp \geq |\alpha_1 \lambda_1^n|_\pp |\Xi|_\pp \\
&\geq |\lambda_1|_\pp^n e^{-h(\alpha_1)D} p^{-v_\pp(\Xi)/e_{\pp}} \\
&\ge |\lambda_1|_\pp^n \, e^{-h(\alpha_1)D}\,
   \exp \Bigl(
     -C(2,D,p)\,\max\{h(\lambda_1/\lambda_2),\log p\}
      \max\{h(\alpha_1/\alpha_2),\log p\}\,\log n\Bigr)
\end{align*}
where $C(2,D,p)$ is as in the statement of Theorem \ref{thm:Yu}, which gives the required form as claimed. 
\end{proof}

We may use this to prove Proposition \ref{prop:3_sum_nonzero}.
\begin{proof}[Proof of Proposition \ref{prop:3_sum_nonzero}]
By Theorem \ref{thm:dominant} there exists an absolute value $|\cdot|_{v'}$ such that (WLOG) 
\begin{align*}
    |\lambda_1|_{v'} > |\lambda_2|_{v'},|\lambda_3|_{v'} \, .
\end{align*}
Now we also have that
\begin{align*}
    |\alpha_1 \lambda_1^n + \alpha_2 \lambda_2^n + \alpha_3 \lambda_3^n|_{v'} &\geq |\alpha_1|_{v'}|\lambda_1|_{v'}^n - |\alpha_2 \lambda_2^n + \alpha_3 \lambda_3^n|_{v'} \\
    &\geq |\alpha_1|_{v'}|\lambda_1|_{v'}^n - (|\alpha_2|_{v'}+ |\alpha_3|_{v'})|\lambda_2|_{v'}^n  \\
    &> 0
\end{align*}
for all $n$ such that
\begin{align} \label{eqn:ngreaterlog}
    n > \frac{\log \left( \frac{|\alpha_2|_{v'} + |\alpha_3|_{v'}}{|\alpha_1|_{v'}} \right)}{\log \left(\frac{|\lambda_1|_{v'}}{|\lambda_2|_{v'}}\right)} \, .
\end{align}
Now we have 
\begin{align} \label{eqn:alphasbound}
    \frac{|\alpha_2|_{v'} + |\alpha_3|_{v'}}{|\alpha_1|_{v'}} \leq \frac{2e^{Dh}}{e^{-Dh}} \leq 2e^{2Dh_\alpha} \, .
\end{align}
In the case where $|\cdot|_{v'}$ is non-Archimedean, corresponding to prime ideal $\pp \subseteq \O_\K$ lying above integer prime $p$, we have 
\begin{align*}
    \frac{|\lambda_1|_{v'}}{|\lambda_2|_{v'}} \geq p^{1/e_\pp} \geq p^{1/D}
\end{align*}
where $e_\pp$ is the ramification index. Therefore,
\begin{align} \label{eqn:lambdaratiobound}
    \log \left(\frac{|\lambda_1|_{v'}}{|\lambda_2|_{v'}}\right) \geq \frac{\log p}{D} \, ,
\end{align}
and \eqref{eqn:ngreaterlog} together with \eqref{eqn:alphasbound} and \eqref{eqn:lambdaratiobound} gives the required bound on $n$.

In the case where $|\cdot|_{v'}$ is Archimedean, note that $\log \left(\frac{|\lambda_1|_{v'}}{|\lambda_2|_{v'}}\right) \neq 0$, and that the degree of $|\lambda_i| = \left(\lambda_i \overline{\lambda_i}\right)^{1/2}$ is at most $2D^2$ for all $\lambda_i$. 

Thus, we apply Theorem \ref{thm:Matveev} to get 
\begin{align*}
    \log \left(\frac{|\lambda_1|_{v'}}{|\lambda_2|_{v'}}\right) \geq \exp(-C D^4 h(|\lambda_1|)h(|\lambda_2|) (1 + \log(2D^2)) \, ,
\end{align*}
for some absolute constant $C>0$. Note that 
\begin{align*}
h(|\lambda_i|) = h\left(\left(\lambda_i \overline{\lambda_i}\right)^{1/2}\right) = \frac{1}{2} h(\lambda_i \overline{\lambda_i}) \leq h(\lambda_i) \leq h_\lambda
\end{align*}
using Proposition \ref{prop:h_props} items \ref{enum:h_props_mult} and \ref{enum:h_props_pow}, which gives
\begin{align*}
 \log \left(\frac{|\lambda_1|_{v'}}{|\lambda_2|_{v'}}\right) \geq  \exp(-\tilde C D^4(1+ \log D)h_\lambda^2) \, ,
\end{align*}
for absolute constant $\tilde C > 0$. Together with \eqref{eqn:ngreaterlog} and \eqref{eqn:alphasbound}, this proves the required bound on $n$. 
\end{proof}
\begin{remark}
Even in this restricted setting, a bound on $n$ that is exponential in $D$ and $h_\lambda$ appears unavoidable with current techniques --- in the Archimedean case one needs a lower bound for $\log(|\lambda_1/\lambda_2|)$, and there exist examples where an exponentially small lower bound is attainable \cite{Dubickas_1998}.
\end{remark}

In order to prove Proposition \ref{prop:3_sum}, we need a lemma on the distance between two points on two circles, an idea credited to Beukers \cite{Beukers1984}. We follow \cite[Section 2.4]{bilu_skolem_2025} closely. Let $\rho_0,\rho_1 \in \RR$ satisfy $0 < \rho_0,\rho_1 \leq 1$, and
\begin{align} \label{eqn:circles_def}
    \mathcal C_0 = \{ z \in \CC : |z| = \rho_0 \} \quad , \quad \mathcal C_1 = \{ z \in \CC : |z-1| = \rho_1 \}
\end{align}
then we have
\begin{align*}
    \mathcal C_0 \cap \mathcal C_1 = \begin{cases}
        \emptyset & \text{when } \rho_0 + \rho_1 < 1 , \\
        \{\rho_0\} & \text{when } \rho_0 + \rho_1 = 1 , \\
        \{\zeta,\overline \zeta\} & \text{when } \rho_0 + \rho_1 > 1 , \\
    \end{cases}
\end{align*}
where
\begin{align} \label{eqn:zeta}
    \zeta = \frac{1 + \rho_0^2 - \rho_1^2 + i\sqrt{4\rho_0^2 \rho_1^2 - (1- \rho_0^2 - \rho_1^2)^2}}{2} \, .
\end{align}
\begin{proposition} \label{prop:two_circles}
Let $z_0 \in \mathcal C_0, \, z_1 \in \mathcal C_1$. Then
\begin{align*}
    |z_0 - z_1| \geq  \begin{cases}
        1 - \rho_0 - \rho_1 & \text{if } \rho_0 + \rho_1 < 1, \\
        \frac{1}{2\rho_0} |z_0 - \rho_0|^2 & \text{if } \rho_0 + \rho_1 = 1, \\
        \frac{2}{3} \left( 1 + \frac{4 \rho_0^2}{(\mathrm{Im} \ \zeta)^2} \right)^{-\frac{1}{2}} \min \{|z_0 - \zeta|, |z_0 - \overline \zeta| \} & \text{if } \rho_0 + \rho_1 > 1 
    \end{cases}
\end{align*}
\end{proposition}
The statement is almost identical to that in \cite{bilu_skolem_2025} and the proof will be almost identical also --- we claim no originality here. The proof will use a short lemma.
\begin{lemma} \label{lem:circle_im}
Let $\eta, \rho > 0$ and let $\mathcal C$ be the circle with centre 0 and radius $\rho$. Then for $z,w \in \mathcal C$ with $|\mathrm{Im}\, z + \mathrm{Im}\, w| \geq \eta$ we have
\begin{align*}
    |z-w| \leq  \left( 1 + \frac{4\rho^2}{\eta^2} \right)^{\frac{1}{2}} |\mathrm{Re}\, z - \mathrm{Re}\, w| \, . 
\end{align*}
\end{lemma}

\begin{proof}
From $(\mathrm{Re}\, z)^2 + (\mathrm{Im}\, z)^2 = (\mathrm{Re}\, w)^2 + (\mathrm{Im}\, w)^2 = \rho^2$, we have
\begin{align*}
|\mathrm{Im}\, z - \mathrm{Im}\, w | &\leq \frac{1}{\eta}|(\mathrm{Im}\, z)^2 - (\mathrm{Im}\, w)^2| \\
&= \frac{1}{\eta}|(\mathrm{Re}\, z)^2 - (\mathrm{Re}\, w)^2| \\ 
&\leq \frac{2\rho}{\eta} |\mathrm{Re} \, z - \mathrm{Re}\, w| \, .
\end{align*}
Therefore, 
\begin{align*}
    |z-w|^2 &= |\mathrm{Re}\, z - \mathrm{Re}\, w|^2 + |\mathrm{Im}\, z - \mathrm{Im}\, w|^2 \\ 
    &\leq \left(1 + \frac{4 \rho^2}{\eta^2} \right) |\mathrm{Re}\, z - \mathrm{Re}\, w|^2 \, ,
\end{align*}
and so
\begin{align*}
    |z-w| \leq \left(1 + \frac{4 \rho^2}{\eta^2} \right)^{\frac{1}{2}} |\mathrm{Re}\, z - \mathrm{Re}\, w|
\end{align*}
as required.
\end{proof}

\begin{proof}[Proof of Proposition \ref{prop:two_circles}]
We consider each case separately. If $\rho_0 + \rho_1 < 1$ then the result is obvious. For the other cases, since
\begin{align} \label{eqn:z0z1rho1}
    |z_0 - z_1| \geq \left| |z_0-1| - |z_1-1| \right|  = \left| |z_0-1| - \rho_1\right| \, ,
\end{align}
it suffices to bound $\left| |z_0-1| - \rho_1\right|$ from below.

Suppose $\rho_0 + \rho_1 = 1$. Then
\begin{align*}
    |z_0-1|^2-\rho_1^2 = 2(\rho_0-\mathrm{Re}\, z_0) = \frac{1}{\rho_0}|z_0-\rho_0|^2 \, .
\end{align*}
Since $|z_0-1|+\rho_1 \leq |z_0|+1 + \rho_1 = 2$, we have
\begin{align*}
    \left| |z_0-1|-\rho_1 \right| \geq \frac{1}{2 \rho_0} |z_0-\rho_0|^2 \, ,
\end{align*}
which together with \eqref{eqn:z0z1rho1} gives what was required in this case.

Finally, suppose $\rho_0 + \rho_1 > 1$. Without loss of generality, assume $\mathrm{Im}\, z_0 \geq 0$ (otherwise consider $\overline{\zeta}$ instead in the subsequent proof). We have
\begin{align*} 
|z_0-1|^2 = \rho_0^2 + 1 - 2 \mathrm{Re}\, z_0, \quad \rho_1^2 = |\zeta-1|^2 = \rho_0^2 + 1 - 2\mathrm{Re}\, \zeta \, ,
\end{align*}
which implies that
\begin{align} \label{eqn:Re1}
    \left| |z_0 - 1|^2- \rho_1^2 \right| = 2 |\mathrm{Re}\, z_0 - \mathrm{Re}\, \zeta | \, .
\end{align}
Now since $\mathrm{Im}\, z_0 \geq 0$ and $\mathrm{Im}\, \zeta > 0$, we may apply Lemma \ref{lem:circle_im} with $\eta := \mathrm{Im}\, \zeta$ to get
\begin{align} \label{eqn:Re2}
    |\mathrm{Re}\, z_0 - \mathrm{Re}\, \zeta| \geq \left(1 + \frac{4 \rho_0^2}{(\mathrm{Im}\, \zeta)^2} \right)^{- \frac{1}{2}} |z_0 - \zeta| \, .
\end{align}
Also, using $\rho_0,\rho_1 \leq 1$, we have $|z_0-1| + \rho_1 \leq |z_0|+1+\rho_1 \leq 3$, which gives
\begin{align} \label{eqn:Re3}
    \left| |z_0-1| - \rho_1 \right| \geq \frac{1}{3} \left| |z_0-1|^2-\rho_1^2 \right| \, .
\end{align}
Putting \cref{eqn:Re1,eqn:Re2,eqn:Re3} together, we get
\begin{align*}
    \left| |z_0-1| - \rho_1 \right| \geq \frac{2}{3} \left(1 + \frac{
    4\rho_0^2}{(\mathrm{Im}\, \zeta)^2} \right)^{- \frac{1}{2}} |z_0- \zeta| \, ,
\end{align*}
which with \eqref{eqn:z0z1rho1} gives what was to be proven.
\end{proof}

With these geometric results in place, we can prove Proposition \ref{prop:3_sum}.

\begin{proof}[Proof of Proposition \ref{prop:3_sum}]
As in the proof of Proposition \ref{prop:2_sum}, assume without loss of generality that $|\cdot|_v = |\cdot|$. Also, by relabelling if necessary, assume that $|\alpha_1| \geq |\alpha_2|,|\alpha_3|$. Let the left-hand side of \eqref{eqn:3_term_nonzero} be denoted by $A(n)$, then we have
\begin{align*}
    A(n) = \alpha_1 \lambda_1^n(z_1-z_0)
\end{align*}
where
\begin{align*}
    z_0 = - \frac{\alpha_2}{\alpha_1} \left( \frac{\lambda_2}{\lambda_1} \right)^n \quad , \quad z_1 = 1 + \frac{\alpha_3}{\alpha_1} \left( \frac{\lambda_3}{\lambda_1} \right)^n \, .
\end{align*}
Set 
\begin{align*}
    \rho_0 = \left| \frac{\alpha_2}{\alpha_1} \right|  \quad , \quad \rho_1 = \left| \frac{\alpha_3}{\alpha_1} \right| \, ,
\end{align*}
noting that $0 \leq \rho_0, \rho_1 \leq 1$ and let $\mathcal C_0, \mathcal C_1$ be as in \eqref{eqn:circles_def}, and $\zeta$ defined as in \eqref{eqn:zeta}. Then $z_0 \in \mathcal C_0, \, z_1 \in \mathcal C_1$ and Proposition \ref{prop:two_circles} gives 
\begin{align} \label{eqn:lower_bound_cases}
    |z_0 - z_1| \geq \begin{cases}
        1 - \rho_0 - \rho_1 & \text{if } \rho_0 + \rho_1 < 1, \\
        \frac{1}{2\rho_0} |z_0 - \rho_0|^2 & \text{if } \rho_0 + \rho_1 = 1, \\
        \frac{2}{3} \left( 1 + \frac{4 \rho_0^2}{(\mathrm{Im} \ \zeta)^2} \right)^{-1/2} \min \{|z_0 - \zeta|, |z_0 - \overline \zeta| \} & \text{if } \rho_0 + \rho_1 > 1 
    \end{cases}
\end{align}
and thus we just need to give lower bounds on the right-hand sides of \eqref{eqn:lower_bound_cases}. For this, note first that the degrees of $\rho_0, \rho_1$ are at most $2D^2$. Secondly, their heights are bounded as follows, using Proposition \ref{prop:h_props}, items \ref{enum:h_props_mult} and \ref{enum:h_props_pow}. For $\rho_0$ we have
\begin{align*}
    h\left( \left| \frac{\alpha_2}{\alpha_1} \right| \right) = h \left( \left(  \frac{\alpha_2}{\alpha_1} \frac{\overline{\alpha_2}}{\overline{\alpha_1}}\right)^{1/2} \right)  = \frac{1}{2} h \left(  \frac{\alpha_2}{\alpha_1} \frac{\overline{\alpha_2}}{\overline{\alpha_1}} \right) \leq 2h_\alpha \, ,
\end{align*}
and the bound for $\rho_1$ is identical.

\underline{Case 1: $\rho_0 + \rho_1 < 1$.} By Proposition \ref{prop:h_props} item \ref{enum:h_props_log} we have
\begin{align*}
    1-\rho_0 - \rho_1 \geq e^{- 2D^2 h(1-\rho_0-\rho_1)} \geq e^{-2D^2(\log 3 + h(\rho_0) + h(\rho_1)) } \geq e^{-\poly(D,h_\alpha)} \, .
\end{align*}

\underline{Case 2: $\rho_0 + \rho_1 = 1$.} Note that
\begin{align*}
    |z_0 - \rho_0| = \left| -\frac{\alpha_2}{\alpha_1} \left( \frac{\lambda_2}{\lambda_1} \right)^n - \left| \frac{\alpha_2}{\alpha_1} \right| \right| \, .
\end{align*}
By Proposition \ref{prop:2_sum}, for $n > \poly(D,h_\alpha)$, we have 
\begin{align*}
    |z_0  - \rho_0| \geq e^{-C_1(1 + \log n)}
\end{align*}
for $C_1 = \poly(D,h_\alpha,h_\lambda)$. 

\underline{Case 3: $\rho_0 + \rho_1 > 1$.} Using Proposition \ref{prop:h_props} it is easy to see that $h(\zeta),h(\mathrm{Im} \ \zeta) = \poly(h_\alpha)$, and the degree of $\zeta$ is $\poly(D)$. Then one applies again Proposition \ref{prop:2_sum} to $|z_0 - \zeta|, |z_0 - \overline \zeta|$ and use \eqref{eqn:lower_bound_cases} to deduce a lower bound of the form 
\begin{align*}
    |z_0 - z_1| \geq e^{-\poly(D,h_\alpha,h_\lambda) \log n} \, .
\end{align*}
for all $n > \poly(D,h_\alpha)$. 

Thus, in every case we have for all $n > \poly(D,h_\alpha)$ that 
\begin{align*}
    |z_0 - z_1| \geq e^{-\poly(D,h_\alpha,h_\lambda) \log n}
\end{align*}
and so
\begin{align*}
    |\alpha_1 \lambda_1^n + \alpha_2 \lambda_2^n + \alpha_3 \lambda_3^n| &\geq |\lambda_1|^n |\alpha_1| e^{-\poly(D,h_\alpha,h_\lambda) \log n} \\
    &\geq |\lambda_1|^n e^{-\poly(D,h_\alpha,h_\lambda) \log n} \, .
\end{align*}
\end{proof}

Before proving our main theorem bounding the zeros of LRS, we require an elementary bound on the heights of roots of polynomials.
\begin{lemma} \label{lem:poly_root_height}
Suppose $\beta \in \Alg$ is a root of a polynomial $P(x) = \sum_{i=0}^m a_i x^i$, for $a_i \in \Alg$. Then $h(\beta) \leq \log 2 + \sum_{i=0}^{m-1} h(a_i/a_m)$.
\end{lemma}
\begin{proof}
We have $a_m \beta^m = -\sum_{i=0}^{m-1} a_i \beta^i$. Using this we derive Cauchy bounds on $|\beta|_v$; define $A_v = \max_{0 \leq i \leq m-1} \left| \frac{a_i}{a_m} \right|_v$. Then in the Archimedean case we have
\begin{align*}
    |a_m|_v|\beta|_v^m \leq \left| \sum_{i=0}^{m-1} a_i \beta^i \right|_v \leq  \left(\max_{0\leq i\leq m-1} |a_i|_v \right)\left| \sum_{i=0}^{m-1} \beta^i \right|_v \leq \left(\max_{0\leq i\leq m-1} |a_i|_v \right) \frac{|\beta|_v^m }{|\beta|_v-1}
\end{align*}
and so rearranging yields $|\beta|_v \leq 1 + A_v \leq 2 \max\{1,A_v \}$. 

In the non-Archimedean case, if we assume $|\beta|_v \geq 1$ then we have
\begin{align*}
    |a_m||\beta|_v^m \leq \left| \sum_{i=0}^{m-1} a_i \beta^i \right|_v \leq |\beta|_v^{m-1} \max_{0 \leq i \leq m-1} |a_i|_v
\end{align*}
using the ultrametric inequality, and hence for any $\beta$ we have $|\beta|_v \leq \max \{ 1, A_v\}$.

Suppose $\beta, a_0, \dots, a_m$ are contained in a number field $\KK$ of degree $D$, then
\begin{align*}
    h(\beta) &= \frac{1}{D} \sum_{v \in M_\KK} D_v \log^+(|\beta|_v) \\
    &\leq \frac{1}{D} \sum_{\substack{v \in M_\KK \\ v \mid \infty}} D_v \log(2 \max\{1, A_v\}) +  \frac{1}{D} \sum_{\substack{v \in M_\KK \\ v \nmid \infty}} D_v \log(\max\{1, A_v\}) \\
    &\leq \frac{1}{D} \sum_{\substack{v \in M_\KK \\ v \mid \infty}} D_v \log 2 + \frac{1}{D} \sum_{v \in M_\KK} D_v \log \left( \prod_{i=0}^{m-1} \max\left\{1,\left|\frac{a_i}{a_m}\right|_v \right\} \right) \\
    &\leq \log 2 + \sum_{i=0}^{m-1} h(a_i/a_m) \, .
\end{align*}
\end{proof}
We finally turn to our main theorem.
\MSTVbound*
\begin{proof}
By the fact that $\LRS{u}$ lies in the MSTV class, there exists an absolute value $|\cdot|_v$ such that, WLOG, 
\begin{align*}
    |\lambda_1|_v = \dots = |\lambda_r|_v > |\lambda_{r+1}|_v \geq \dots \geq |\lambda_s|_v \, ,
\end{align*}
where $r \leq 2$ if $|\cdot|_v$ is non-Archimedean and $r \leq 3$ if $|\cdot|_v$ is Archimedean. If $|\cdot|_v$ is non-Archimedean, let $p \in \N$ be the prime such that $v \mid p$. We have
\begin{align} \label{eqn:un_lower_bound}
|u_n|_v &\geq \underbrace{\left| \sum_{i=1}^r P_i(n) \lambda_i^n \right|_v}_{(*)} - \underbrace{\left| \sum_{i=r+1}^s P_i(n) \lambda_i^n \right|_v}_{(**)} \, .
\end{align}
First we upper bound $(**)$. Using \cref{prop:h_props} repeatedly we may bound $h(P_i(n)) \leq \poly(h_P,\delta)\log n$ and hence by \cref{prop:h_props} \cref{enum:h_props_log} we get
\begin{align} \label{eqn:double_ast_bound}
     \left| \sum_{i=r+1}^s P_i(n) \lambda_i^n \right|_v \leq |\lambda_{r+1}|^nse^{D \poly(h_P,\delta) \log n} \, .
\end{align}
To lower bound $(*)$, we first note that by \cref{lem:poly_root_height}, and using \cref{prop:h_props}, $P_i(n) = 0$ implies $n \leq e^{\poly(h_P,\delta)}$. So for $n > e^{\poly(h_P,\delta)}$, we have $P_i(n) \neq 0$ for all $i$, and we either apply Proposition \ref{prop:2_sum} or \cref{prop:3_sum} with Proposition \ref{prop:3_sum_nonzero} when $r=2$ or $3$ respectively, to get that for all $n > e^{\poly(D,h_\lambda,h_P,\delta)}$
\begin{align}
\left| \sum_{i=1}^r P_i(n) \lambda_i^n \right|_v &\geq |\lambda_1|_v^n e^{-\poly \left(p^D,D,h_\lambda,\max_{1 \leq i \leq r} \{h(P_i(n))\} \right)\log n} \notag \\
&\geq |\lambda_1|_v^n e^{-\poly \left(p^D,D,h_\lambda ,h_P,\delta\right)(\log n)^2} \, . \label{eqn:ast_bound}
\end{align}
Note now that $|u_n|_v > 0$ and hence $u_n \neq 0$ whenever $(*) > (**)$, which by 
\cref{eqn:un_lower_bound,eqn:ast_bound,eqn:double_ast_bound} 
occurs when $n > e^{\poly(D,h_\lambda,h_P,\delta)}$ and
\begin{align} \label{eqn:lambda1_dominates}
|\lambda_1|_v^n e^{-\poly \left(p^D,D,h_\lambda,h_P,\delta \right)(\log n)^2} \geq |\lambda_{r+1}|^nse^{D \poly(h_P,\delta) \log n}
\end{align}
which holds if
\begin{align} \label{eqn:n_logn_bound}
n > \frac{\poly(p^D,D,h_\lambda,h_P,\delta,s) (\log n)^2}{\log \left| \frac{\lambda_1}{\lambda_{r+1}} \right|_v} \, .
\end{align}
Note that $\log n \leq 2n^{1/3}$ for all $n \geq 1$, so certainly \eqref{eqn:n_logn_bound} holds for
\begin{align*}
n > \frac{\poly(p^D,D,h_\lambda,h_P,\delta,s)}{\left( \log \left| \frac{\lambda_1}{\lambda_{r+1}} \right|_v \right)^3} \, .
\end{align*}
Also, as in the proof of \cref{prop:3_sum_nonzero} (by applying \cref{thm:Matveev} or using discreteness of non-Archimedean absolute values), $\log \left| \frac{\lambda_1}{\lambda_{r+1}} \right|_v \geq e^{-\poly(D,h_\lambda)}$. So we have $u_n \neq 0$ for $n > \poly(p^D,s) e^{\poly(D,h_\lambda,h_P,\delta)}$. Finally, note that in the non-Archimedean case, $|\cdot|_v = |\cdot|_\pp$ for prime ideal $\pp \subseteq \O_\KK$, and the non-Archimedean bound is only invoked when $v_\pp(\lambda_i) > 0 $ for some $i$. We also have by Proposition \ref{prop:h_props} that
\begin{align*}
    p^{-1/D} \geq p^{-v_\pp(\lambda_i)/e_\pp} = |\lambda_i|_v \geq e^{-Dh_\lambda}  
\end{align*}
so $p \leq e^{D^2h_\lambda}$, so we finally have $u_n \neq 0$ for 
\begin{align*}
    n > \poly(s)e^{\poly(D,h_\lambda,h_P,\delta)} \, .
\end{align*}%
\end{proof}

\newpage
\section{Entry (8,6) from~\texorpdfstring{\cref{table}}{Figure 1}} \label{sec:table_hard}
We justify here that $\orbprob(8,6)$ --- equivalently, $\simskol(8,2)$ --- is $\skolem_\Q(5)$-hard. 
\begin{theorem} \label{thm:skolem5_hardness}
	If $\simskol(8,2)$ is decidable,
	then $\skolem_\Q(5)$ is decidable.
\end{theorem}
\begin{proof}
	By \cite{Bilu_2023} the hard cases of the Skolem Problem on 
	$\Q$-LRS may be reduced to non-degenerate LRS of the form
	\begin{align*}
		u_n = \alpha_1 \lambda_1^n + \overline{\alpha}_1 \overline{\lambda}_1^n + \alpha_2 \lambda_2^n + \overline{\alpha}_2 \overline{\lambda}_2^n + \alpha_3 \lambda_3^n
	\end{align*}
	where $|\lambda_1| = |\lambda_2| > |\lambda_3|$ and $|\alpha_1| \neq |\alpha_2|$. Note this means that $\lambda_1 \overline{\lambda}_1 = \lambda_2 \overline{\lambda}_2$. 
	
	Now consider two LRS defined as
	\begin{align*}\LRS{u}^{(1)}_n \coloneq&\,
		\left( \frac{\lambda_1^n}{\lambda_2^n} - \frac{\alpha_2}{\alpha_1} \right) \left( \alpha_1 \lambda_1^n + \overline{\alpha}_1 \overline{\lambda}_1^n + \alpha_2 \lambda_2^n + \overline{\alpha}_2 \overline{\lambda}_2^n + \alpha_3 \lambda_3^n \right) \\
		=&\, \alpha_1 \frac{\lambda_1^{2n}}{\lambda_2^n} + \overline{\alpha}_2 \frac{\lambda_1^n \overline{\lambda}_2^n}{\lambda_2^n} + \alpha_3 \frac{\lambda_1^n\lambda_3^n}{\lambda_2^n} + \\
        &\left( \overline{\alpha}_1 - \frac{\alpha_2 \overline{\alpha_2}}{\alpha_1} \right)\overline{\lambda}_2^n 
        - \frac{\alpha_2 \overline{\alpha}_1}{\alpha_1}\overline{\lambda}_1^n - \frac{\alpha_2^2}{\alpha_1} \lambda_2^n - \frac{\alpha_2 \alpha_3}{\alpha_1} \lambda_3^n
	\end{align*}
	and 
	\begin{align*}
		\LRS{u}^{(2)}_n \coloneq&\,
		\left( \frac{\lambda_1^n}{\lambda_2^n} - \frac{\overline{\alpha}_1}{\overline{\alpha}_2} \right) \left( \alpha_1 \lambda_1^n + \overline{\alpha}_1 \overline{\lambda}_1^n + \alpha_2 \lambda_2^n + \overline{\alpha}_2 \overline{\lambda}_2^n + \alpha_3 \lambda_3^n \right)\\
		=&\, \alpha_1 \frac{\lambda_1^{2n}}{\lambda_2^n} + \alpha_3 \frac{\lambda_1^n\lambda_3^n}{\lambda_2^n} +\left(\alpha_2 - \frac{\alpha_1 \overline{\alpha_1}}{\overline{\alpha}_2} \right)\lambda_1^n + \\
        &\overline{\alpha}_2 \frac{\lambda_1^n \overline{\lambda}_2^n}{\lambda_2^n} 
        - \frac{ \overline{\alpha}_1^2}{\overline{\alpha}_2}\overline{\lambda}_1^n - \frac{\overline{\alpha}_1 \alpha_2}{\overline{\alpha}_2} \lambda_2^n - \frac{\overline{\alpha}_1 \alpha_3}{\overline{\alpha}_2} \lambda_3^n \, .
	\end{align*}
	Both $\LRS{u}^{(1)}$ and $\LRS{u}^{(2)}$ are $\Alg$-LRS that satisfy the same recurrence relation~of order~8 
	as can be observed from them sharing the same set of characteristic roots
	\[\left\{\frac{\lambda_1^2}{\lambda_2}, \frac{\lambda_1\overline{\lambda}_2}{\lambda_2},\frac{\lambda_1\lambda_3}{\lambda_2}, \lambda_1,\overline{\lambda}_1,\lambda_2,\overline{\lambda}_2,\lambda_3\right\}.\]
	The two sequences are linearly independent, as the factors $\frac{\lambda_1^n}{\lambda_2^n} - \frac{\alpha_2}{\alpha_1}$ and $\frac{\lambda_1^n}{\lambda_2^n} - \frac{\overline{\alpha_1}}{\overline{\alpha_2}}$ are linearly independent from the fact $|\alpha_1| \neq |\alpha_2|$.
	Moreover, 
	$\mathcal{Z}(u^{(1)},u^{(2)}) = Z(u)$
	since  $|\alpha_1| \neq |\alpha_2|$. Therefore, if the $\simskol(8,2)$ was decidable, one could determine $\mathcal{Z}(u^{(1)},u^{(2)})$ and thus $\mathcal{Z}(u)$. Since $u$ was chosen arbitrary, this shows the Skolem Problem for $\Q$-LRS of order 5 would be decidable, and hence $\skolem_\Q(5)$ would be.
\end{proof}

Note that using the ideas of \cref{thm:skolem5_hardness} and \cref{sec:hardness}, one can prove $\skolem_\KK(k)$-hardness for $\KK = \Q$ or $\KK = \Alg$ for many other intermediate values of $d$ and $t$, past the obvious reduction when $t = d-1$.

\end{document}